\documentclass[11pt]{article}
\usepackage[margin=1in]{geometry}
\usepackage{amsmath,amssymb,amsthm,mathtools}
\usepackage{microtype,enumitem,xcolor}
\usepackage[colorlinks=true,linkcolor=blue!45!black,citecolor=blue!45!black,urlcolor=blue!45!black]{hyperref}
\usepackage{bookmark}
\usepackage{cleveref}
\hypersetup{pdftitle={Partially Dynamic Single-Source Shortest Paths in Almost-Linear Time}}
\numberwithin{equation}{section}
\newtheorem{theorem}{Theorem}[section]
\newtheorem{lemma}[theorem]{Lemma}

\newcommand{\eps}{\epsilon}
\DeclareMathOperator{\dist}{dist}
\DeclareMathOperator{\cut}{cut}

\usepackage{xcolor}

\usepackage[T1]{fontenc}
\usepackage{lmodern}

\title{Partially Dynamic Single-Source Shortest Paths\\ in Almost-Linear Time }
\author{Bin Cui\thanks{Peking University. Email: \texttt{bin.cui@pku.edu.cn}.}
\and Hanqing Li\thanks{Peking University. Email: \texttt{2400012952@stu.pku.edu.cn}.}
\and Xupeng Miao\thanks{Peking University. Email: \texttt{xupeng.miao@pku.edu.cn}.}
\and Haoze Wang\thanks{Unaffiliated. Email: \texttt{w97871749@gmail.com}}
}
\date{}
\begin{document}
\maketitle
\begin{abstract}
We study $(1+\epsilon)$-approximate single-source distances in
directed graphs with positive integer edge weights in $[1,W]$.
For every $0<\epsilon\le1$, we give deterministic data structures
for either edge insertions or edge deletions with total update time
$m^{1+o(1)}\epsilon^{-1}\log W$, where $m\ge n$ bounds the
total number of initially present and inserted edges.
Both structures explicitly maintain approximate distance estimates,
answer distance queries in constant time, and identify unreachable
vertices exactly. The incremental structure additionally reports
a simple $(1+\epsilon)$-approximate source path $P$ in
$O(|P|)$ time.

Our decremental structure gives the first almost-linear-time algorithm.
It improves the previous best randomized expected bound of
$\widetilde O(\min\{n^2,mn^{2/3}\})$ against an oblivious
adversary~\cite{BPW20}.
Our incremental bound improves the $\epsilon^{-B}$ dependence
of the $B$-level recursive algorithm of Liu~\cite{Liu26}
to an explicit $\epsilon^{-1}$ factor while retaining
almost-linear dependence on $m$. 

Our algorithms use modified interior point methods.
For edge deletions, we extend the single-pair distance framework
of~\cite{vBCKLMPSS24} to all targets simultaneously.
A shared potential and one dynamic minimum-ratio cut instance
support global updates to the distance labels.
Our main insight is that constant relative cut balance certifies approximate distances
at every target.
For edge insertions, we develop an analogous circulation-based
formulation implemented through the dynamic minimum-ratio cycle
structure of~\cite{CKLMP24}.

We also give a black-box reduction from distance estimates to
path reporting in incremental graphs, with
$m^{1+o(1)}\epsilon^{-1}\log W$ additional maintenance time.
This reduction may be of independent interest.
\end{abstract}

\newpage
\tableofcontents
\newpage
\section{Introduction}
\label{sec:introduction}

In this paper we study the single-source shortest path problem (SSSP) in dynamic directed graphs with positive edge
lengths. Specifically, we focus on \emph{incremental} and \emph{decremental} graphs, where edges only undergo insertions and deletions, respectively.
Our goal is to design an algorithm which maintains $(1+\epsilon)$-approximate distances from a source vertex $s$ to all other vertices.

Over the last several decades, there has been a vast literature on SSSP in incremental and decremental graphs~\cite{ES81,DHZ00,RZ11,HKN14,HKN15,HKN16,Ber16,BPW20,KL19,KL20,PVW20,PW20,CZ21,KMP22,GK25,Liu26}.
One motivation is that in both the
directed and undirected settings, a fast data structure for decremental SSSP with path reporting,
combined with the multiplicative weights update (MWU) framework, gives algorithms for maximum flow, multicommodity flow, and several other problems~\cite{GK98,Madry10,BPS20,BPS21,CK24b,CK24a}.
In particular, approximate
shortest paths suffice for these applications, and additionally, there are conditional lower bounds
against exact shortest paths in both incremental and decremental graphs~\cite{AV14,HKNS15,AHRVW19,JX22,SVXY25}.
However, until recently, almost-linear-time algorithms for SSSP in decremental graphs were still unknown. The work~\cite{Liu26} gives an almost-linear-time algorithm for SSSP in incremental graphs, with $\epsilon^{-B}$ dependence, where $B$ is the number of levels.

Building off the almost-linear-time mincost flow algorithm of~\cite{CKLPPS22}, the works~\cite{vBLS23,CKLMP24,vBCKLMPSS24}
gave a framework based on a dynamic interior point method (IPM) to solve
dynamic minimum cost flow in incremental and decremental graphs. Thus, they achieved almost-linear-time
algorithms for approximate $s$-$t$ mincost flow, which in turn give algorithms for single-pair
shortest paths, maxflow, cycle detection, and more. 
Based on this framework, an almost-linear-time approximate SSSP data structure in incremental directed graphs~\cite{Liu26} can be obtained by using several modifications of the dynamic IPM driving previous works.

Our work also builds on IPM. We give almost-linear-time approximate SSSP data structures for both incremental and decremental directed graphs. In the decremental setting, we maintain only distance estimates, whereas in the incremental setting, we also support path reporting.
\begin{theorem}
\label{thm:main}
Let $G$ be a decremental graph undergoing $m$ edge deletions
with integral edge lengths in $[1,W]$ and a fixed source $s \in V$.
There is a deterministic data structure that explicitly maintains
distance estimates $\widetilde{d} \colon V \to \mathbb{R}_{\ge 0}\cup\{\infty\}$ satisfying
\[
    d_G(s,v) \le \widetilde{d}(v) \le (1+\epsilon)d_G(s,v),
\]
in total update time $m^{1+o(1)}\epsilon^{-1}\log W$.
\end{theorem}
  
\begin{theorem}
\label{thm:main-2}
Let $G$ be an incremental graph undergoing $m$ edge insertions
with integral edge lengths in $[1,W]$ and a fixed source $s \in V$.
There is a deterministic data structure that explicitly maintains
distance estimates $\widetilde{d} \colon V \to \mathbb{R}_{\ge 0}\cup\{\infty\}$ satisfying
\[
    d_G(s,v) \le \widetilde{d}(v) \le (1+\epsilon)d_G(s,v),
\]
and reports approximate shortest paths $\pi_{s,v}$
from $s$ to $v$ in time $O(|\pi_{s,v}|)$.
The total running time of the data structure is $m^{1+o(1)}\epsilon^{-1}\log W$.
\end{theorem}

The previous best runtime was $\widetilde O(\min\{n^2,mn^{2/3}\})$ against an oblivious
adversary for decremental graphs~\cite{BPW20}, and $m^{1+1/B+o(B)}\epsilon^{-B}\log W$ 
for arbitrary $B\in \mathbb N^+$ for incremental graphs~\cite{Liu26}.
As a corollary,
Theorem~\ref{thm:main-2} also gives a deterministic
APSP data structure in incremental graphs with running time $m^{1+o(1)}n\epsilon^{-1}\log W$.

Our black-box reduction from distance estimates to path reporting in incremental directed graphs may be of independent interest. Given explicitly maintained distance estimates with accuracy parameter \(\delta=\Theta(\epsilon/\log m)\), the reduction supports reporting simple \((1+\epsilon)\)-approximate shortest paths in \(O(|P|)\) time per reported path \(P\). It uses \(O(m+n)\) additional space and \(O(m\epsilon^{-1}(\log m)^2\log n\log(nW))\) additional total maintenance time. The reduction interacts with the underlying distance algorithm only through changes to its estimates, making it applicable to other incremental distance-estimation algorithms satisfying this interface.

\paragraph{Preliminaries.}
We sometimes say that a graph $G$ is incremental to mean that it represents a
sequence of graphs $G^{(0)},G^{(1)},\ldots$, where each graph has an extra edge compared to the previous one. 
We sometimes say that a graph $G$ is decremental to mean that it represents a
sequence of graphs $G^{(0)},G^{(1)},\ldots$, where each graph has one fewer edge compared to the previous one. 
We use $\tilde{O}$ and $\hat{O}$-notation to suppress logarithmic and subpolynomial factors, respectively. 
We let $B\in\mathbb{R}^{E\times V}$ denote the edge-vertex incidence matrix of a graph $G$.
We say that a flow
$f\in\mathbb{R}^{E}$ routes a demand $d\in\mathbb{R}^{V}$ if $B^{\top}f=d$.
A flow routing demand $0$ is called a circulation.

In the remainder of this paper, we assume that the maximum edge length $W$ and the inverse accuracy parameter $\epsilon^{-1}$ are polynomially bounded, as there is a standard way to go from that case to the general one while paying an extra $\log W$ factor\footnote{For instance, one can consider a distance scale $[L,2L]$, and increase all small edge lengths to at least $L/m^2$, and delete all edges of length greater than $2L$, and run the algorithm for each $L = 2^i$ for $i \le O(\log (mW))$.}, and we can simply recompute the whole graph after each update if $\epsilon^{-1}=\Omega(m)$.

\paragraph{Organization.}
Section~\ref{sec:overview} gives an overview of our techniques.
Section~\ref{sec:decremental-sssp} presents and analyzes the decremental distance-estimation
algorithm.
Section~\ref{sec4:incremental} develops the incremental distance-estimation algorithm.
Section~\ref{sec5:paths} gives a black-box reduction from incremental distance
estimates to path reporting.
\section{Technical Overview}
\label{sec:overview}

\subsection{Prior work on interior point methods}
\label{sec2:prior-ipm}
Our algorithms build on the recent almost-linear-time algorithms
for decremental single-pair distances~\cite{vBCKLMPSS24} and incremental SSSP~\cite{Liu26}. 

\begin{theorem}[Decremental single-pair distances,
    {\cite{vBCKLMPSS24}}]
\label{sec2:prior-single-pair}
Let $G$ be a directed graph with $m$ initial edges, positive integer
weights bounded by $m^{O(1)}$, and fixed vertices $s,t$.
For every $0<\epsilon\leq1$, a deterministic data structure maintains a
$(1+\epsilon)$-approximation to the $s$--$t$ distance under edge deletions
in total update time $m^{1+o(1)}\epsilon^{-1}$.
\end{theorem}

This algorithm uses a dual IPM implemented through dynamic minimum-ratio
cuts to estimate the distance of one prescribed pair.
Intuitively, the algorithm maintains $p\in\mathbb R^V$, and seeks to maximize $p_t$ subject to $p_s=0$ and $p_v-p_u\leq w(u,v)$.
Edge deletions only remove dual constraints, so the current vertex
potentials remain feasible. For each fixed objective threshold, the dual
$\ell_1$-IPM makes progress by shifting potentials on one side of an
approximately minimum-ratio cut, which provides a descent direction for
the IPM potential function. These cuts are found efficiently using a
dynamic tree-cut sparsifier maintained via a capacitated expander
hierarchy. 

An aggregate objective guarantee does not by itself give the
simultaneous guarantees required by SSSP.
Indeed, one can encode all
source distances in a single optimization problem: maximizing
$\sum_{v\ne s}p_v$ subject to $p_s=0$ and
$p_v-p_u\leq w(u,v)$ has the distance vector as its optimum when every
vertex is reachable. However, a small relative error in this sum can
hide a large relative error at a vertex whose distance is small.

\begin{theorem}[Incremental SSSP, {\cite{Liu26}}]
\label{sec2:prior-incremental}
Let $G$ undergo $m$ edge insertions with positive integer weights in
$[1,W]$ and fixed source $s$.
For any positive integer $B$, a deterministic
data structure explicitly maintains $(1+\epsilon)$-approximate
source distances in total update time $m^{1+1/B+o(B)}\epsilon^{-B}\log W$
and reports an approximate source--target path $P$ in time $O(|P|)$.
\end{theorem}
The algorithm uses $B$ recursive levels.
At each level, a modified primal IPM, implemented through dynamic
minimum-ratio cycles, monitors flows on auxiliary source-to-vertex
edges to detect vertices whose stored distances may have become
inaccurate.
A high-power barrier on these edges bounds the total number of
vertices flagged between rebuilds.
Unflagged vertices retain valid estimates, while flagged vertices
are handled recursively in a smaller contracted graph.
A sufficiently large decrease in the sum of distances at any monitored
distance scale triggers a rebuild and an exact SSSP recomputation.
The recursive analysis incurs an $\epsilon^{-1}$ factor per level,
yielding the $\epsilon^{-B}$ dependence; thus, taking $B$ to grow with
$m$ to obtain almost-linear update time does not give a fixed
polynomial dependence on $\epsilon^{-1}$.


These results motivate a stronger optimization certificate: approximately optimizing a shared potential should certify relative accuracy
at every target simultaneously.

\subsection{Our techniques}
\label{sec2:our-techniques}

We obtain an all-target certificate using indegree-weighted
artificial return edges and amplified logarithmic barriers.
Auxiliary edges and strengthened barriers also appear in \cite{Liu26}.
We first describe the decremental setting; the incremental algorithm
uses a related circulation-based formulation.

\paragraph{Decremental SSSP.}
Let $\widehat G=(V,E_{\rm core},c)$ be the normalized core graph.
Every vertex is reachable from $s$, and no core edge enters $s$.
Deleted edges are retained at a sufficiently large cost, so the core
edge set remains fixed.
For each $v\ne s$, add a zero-cost return edge $r_v=v\to s$, and let
$q(v)$ denote the indegree of $v$ in the core graph.
For $p\in\mathbb R^V$ with $p_s=0$, define
\begin{equation}\label{sec2:overview-potential}
 \Psi(p)=
 \underbrace{\sum_{e=(u,v)\in E_{\rm core}}
   \psi_\kappa(c_e+p_u-p_v)}_{\text{core barriers}}
 +
 \underbrace{A\sum_{v\ne s}q(v)\psi_\kappa(p_v)}_
   {\text{amplified return barriers}},
\end{equation}
where $A=\Theta(\epsilon^{-1}\log m)$ is a sufficiently large
constant multiple, $\kappa>0$ is small, and
\[
 \psi_\kappa(t)=
 \begin{cases}
  -\log t, & t\geq\kappa,\\
  -\log\kappa+1-t/\kappa, & t<\kappa.
 \end{cases}
\]

For a core edge $e=u\to v$, the slack $c_e+p_u-p_v$ corresponds
to the constraint $p_v-p_u\leq c_e$.
Its logarithmic barrier resists increasing $p_v$ as the slack
approaches zero.
Clipping keeps the potential defined during repairs that temporarily
violate a constraint.
For a return edge, the slack is simply $p_v$.
In the unclipped region, its contribution is $-Aq(v)\log p_v$,
whose derivative $-Aq(v)/p_v$ pushes the label upward.
The factor $q(v)$ matches the number of incoming core constraints,
while $A$ controls the strength of this upward pressure.
For example, if the core graph consists of a single edge $s\to v$
of cost $c$, minimizing the unclipped potential
$-\log(c-p_v)-A\log p_v$ gives $p_v=Ac/(A+1)$.
Our all-target argument shows how this mechanism operates
simultaneously across the coupled constraints of the graph.

To describe the certificate, assign each logical edge $e=(u,v)$
the capacity
\[
 f_e=\frac{\mu_e}{\max\{c_e+p_u-p_v,\kappa\}},
 \qquad
 \mu_e=
 \begin{cases}
  1, & e\in E_{\rm core},\\
  Aq(v), & e=r_v.
 \end{cases}
\]
Let $B$ be the tail-minus-head incidence matrix.
Then $\nabla\Psi(p)=-B^\top f$.
For any cut $C\subseteq V$, let $I(C)$ and $O(C)$ denote the total
capacities entering and leaving $C$, respectively.
Thus
\[
 \nabla\Psi(p)^\top\mathbf1_C=I(C)-O(C).
\]
We use the dynamic minimum-ratio cut structure of \cite{vBCKLMPSS24}
to find descent directions until its stopping condition is satisfied.
The resulting certificate implies
\[
 |I(C)-O(C)|\leq\tfrac13\bigl(I(C)+O(C)\bigr)
 \qquad(C\subseteq V).
\]
We show that this constant-factor cut balance implies
\begin{equation}\label{sec2:all-target-certificate}
 p_v\leq \operatorname{dist}_{\widehat G}(s,v)
 \leq(1+\epsilon/4)p_v
 \qquad(v\ne s).
\end{equation}

We bound the initial value of $\Psi$ and establish a global lower
bound.
Increasing the cost of a deleted edge cannot increase $\Psi$,
while each successful descent decreases it by at least
$m^{-o(1)}$.
Together, these bounds control the total number of descents and
yield the stated update time.
For insertions, we use a related circulation-based potential
framework with threshold probes, implemented through the dynamic
minimum-ratio cycle structure of \cite{CKLMP24}.

\paragraph{The black-box reduction.}
We maintain a distance label for each vertex and a parent pointer
for each reachable nonsource vertex.
Path queries follow parent pointers from the target back to the source.
To maintain these pointers efficiently, we organize vertices into
$O(\log n)$ levels.
After each batch of oracle reports, we remove the reported vertices
from their current buckets and choose the lowest empty level $k$.
Let $S$ consist of the reported vertices and all vertices at levels
below $k$.
We perform a local Dijkstra repair on $S$, using labels outside $S$
as boundary values, and assign every vertex in $S$ to level $k$.
The approximation loss is at most a factor of $1+O(\delta)$ per level,
giving a total approximation factor of
$(1+\delta)^{O(\log n)}$.
Thus choosing $\delta=\Theta(\epsilon/\log n)$ with a sufficiently
small constant gives $(1+\epsilon)$-approximate paths.

\section{Decremental SSSP}
\label{sec:decremental-sssp}
The goal of this section is to establish the following theorem:
\begin{theorem}
\label{thm:decremental}
Let $m\ge n \ge 2$, $W = m^{O(1)}$ and $1/m < \epsilon < 1$.
Let $G$ be an $n$-vertex decremental graph undergoing $m$ edge deletions
with integral edge lengths in the range $[1,W]$ with a source $s \in V$.
There is a deterministic data structure that explicitly maintains
distance estimates $D:V\to\mathbb R_{\ge0}\cup\{\infty\}$
\[
    d_G(s,v) \le D(v) \le (1+\epsilon)d_G(s,v),
\]
in total update time $m^{1+o(1)}\epsilon^{-1}$.
\end{theorem}

The structure of this section is organized as follows:
In \Cref{sec:setup}, we give constants and variables used in the algorithm.
In \Cref{subsec:main-algorithm}, we introduce the minimum-ratio cut interface and our algorithm.
In \Cref{sec:algorithm-analysis}, we prove the correctness of the algorithm.
Finally in \Cref{sec:running-time}, we bound the running time.

\subsection{Basic Notation}
\label{sec:setup}

Let $G_0=(V,E_0,w_0)$ be the initial graph, and let
$G_i=(V,E_i,w_i)$ be the graph after $i$ updates.
Without loss of generality, we assume that every vertex is initially reachable from $s$,
and there exists no edge entering $s$.
Recall that all original edge weights lie in $[1,W]$. Set
\begin{equation}
 H=4nW,\qquad Q=8nH.
 \label{eq:normalization}
\end{equation}

Normalize every edge weight by $Q$ to obtain
\[
 \widehat G_0=(V,E_0,c_0),
 \qquad c_0(e)=w_0(e)/Q.
\]
We call the edges in $E_0$ \emph{core edges}.

We maintain $\widehat G_i$ synchronously with $G_i$, but retain
every core edge throughout the sequence. 
Whenever an edge $e$
is deleted from the original graph, increase its cost in the
normalized graph to $H/Q$ instead of removing it.
Thus
\[
 \widehat G_i=(V,E_0,c_i),
 \qquad
 c_i(e)=
 \begin{cases}
  w_i(e)/Q,&e\in E_i,\\
  H/Q,&e\in E_0\setminus E_i.
 \end{cases}
\]
Every core cost satisfies
\begin{equation}
 \frac1Q\le c_i(e)\le\frac1{8n}.
 \label{eq:core-costs}
\end{equation}

Write $d_i(v)=\dist_{G_i}(s,v)$ and $z_i(v)=\dist_{\widehat G_i}(s,v)$.
For every $v\ne s$, we have
\begin{equation}\label{eq:distance-gap}
 \begin{cases}
   Qz_i(v)= d_i(v)<H/4,& d_i(v)<\infty,\\
   H \le Qz_i(v) \le (n-1)H,& d_i(v)=\infty.
 \end{cases}
\end{equation}
Indeed, every finite distance in the original graph is attained by a simple path of weight at
most $(n-1)W$, whereas any core path using a deleted edge has normalized cost at least $H/Q$.
In particular, all core vertices are reachable and
$1/Q\leq z_i(v)<1/8$ for $v\ne s$.



For the optimization problem, introduce a distinct artificial
\emph{return edge} $r_v=v\to s$ of cost zero for every $v\ne s$.
Let
\[
 E_{\mathrm{ret}}=\{r_v:v\ne s\},
 \qquad
 \mathcal E=E_0\cup E_{\mathrm{ret}}.
\]
We use the term \emph{logical edge} to denote both core edges and return edges, i.e. edges in $\mathcal E$.

Now we define some constants and variables for the algorithm.
Choose
\begin{equation}
 A=\left\lceil\frac{32\log_2(m+1)}{\eps}\right\rceil,\quad \kappa=(AmQ)^{-2}.
 \label{eq:parameters}
\end{equation}

Define
\begin{equation}
 q(v)=\deg^-_{G_0}(v)\quad(v\ne s),\qquad q(s)=0.
 \label{eq:return-mass}
\end{equation}
Initial source reachability gives $q(v)\ge1$ for $v\ne s$. We have
\begin{equation}
 \sum_{v\ne s}q(v)\le m.
 \label{eq:fixed-indegrees}
\end{equation}

For each logical edge, define its coefficient by
\begin{equation}
 \mu_e=
 \begin{cases}
  1,&e\in E_0,\\
  Aq(v),&e=r_v\in E_{\mathrm{ret}}.
 \end{cases}
 \label{eq:mu-definition}
\end{equation}
Every $\mu_e\ge1$, the coefficients are fixed, and $\sum_{e\in\mathcal E}\mu_e\le m+Am\le2Am$.

Define the clipped logarithm
\[
 \psi_\kappa(t)=
 \begin{cases}
  -\log t,&t\ge\kappa,\\
  -\log\kappa+1-t/\kappa,&t\le\kappa.
 \end{cases}
\]
It is convex and continuously differentiable, with
$\psi_\kappa'(t)=-1/\max\{t,\kappa\}$. 

For a label vector $p\in\mathbb R^V$ with $p_s=0$,
define the label potential as
\begin{equation}
 \Psi(p)=\sum_{e\in\mathcal E}\mu_e\psi_\kappa(\sigma_e), 
 \label{eq:potential}
\end{equation}
and define the weighted norm as
$$
 \|x\|_f=\sum_{e\in\mathcal E}f_e|(Bx)_e|.
$$
For any logical edge $e=u\to v$, define its slack, clipped denominator, and capacity by
\begin{equation}
 \sigma_e=c_e+p_u-p_v,\qquad
 t_e=\max\{\sigma_e,\kappa\},\qquad
 f_e=\frac{\mu_e}{t_e}.
 \label{eq:slack-capacity}
\end{equation}

We next define the incident capacities:
\begin{equation}
 I(v)=\sum_{e=u\to v\in\mathcal E}f_e,\qquad
 O(v)=\sum_{e=v\to u\in\mathcal E}f_e.
 \label{eq:incident-capacities}
\end{equation}
Differentiating the potential function gives:
\begin{equation}
 g=\nabla\Psi=-B^{\mathsf T}f,
 \qquad g_v=I(v)-O(v)\quad(v\in V).
 \label{eq:gradient}
\end{equation}
This formula includes the source coordinate, and $\sum_vg_v=0$. The capacities are generated by the potential; no flow-feasibility assumption is imposed on them.

For any set $R\subseteq V$, define $I(R)$ as the total capacity of all
logical edges entering $R$ from $V\setminus R$, and $O(R)$ as the total
capacity of all logical edges leaving $R$ for $V\setminus R$.
Summing vertex gradients cancels internal edges, so
\[
 g\cdot\mathbf1_R=I(R)-O(R),\qquad
 \|\mathbf1_R\|_f=I(R)+O(R).
\]

\subsection{Algorithm}
\label{subsec:main-algorithm}
\subsubsection{The dynamic cut interface}
On a connected \emph{undirected} multigraph with capacities $u_e$, for $x\in\mathbb R^V$, define
\[
 \|x\|_u=\sum_{e=\{a,b\}}u_e|x_a-x_b|.
\]
For a vector $\gamma\in\mathbb R^V$ with $\sum_v\gamma_v=0$ and $|\gamma_v|\le U_0$, define
\[
 r^*=\min_{\|x\|_u>0}\frac{\gamma\cdot x}{\|x\|_u}.
\]
For $C\subseteq V$, write $\gamma(C)=\sum_{v\in C}\gamma_v$ and let $\cut_u(C)$ be the sum of capacities crossing $C$.

\begin{lemma}[Adapted from {\cite[Theorem~5.4]{vBCKLMPSS24}}]
\label{prop:cut-interface}
For a graph-size bound $m$, a deterministic structure supports edge insertions and deletions, paired changes to $\gamma$,\footnote{For an entry pair $(u,v)$, update $\gamma_u = \gamma_u+\delta$ and $\gamma_v = \gamma_v-\delta$.} and queries of implicit potentials $y_v$, initially zero. When $\gamma\ne0$, an interface update returns $G<0$, $\mathcal U>0$, and an implicit cut $C$ satisfying
\begin{equation}
 G=\gamma(C),\qquad \cut_u(C)\le\mathcal U,
 \qquad \frac G{\mathcal U}\le\frac{r^*}{K}.
 \label{eq:cut-interface}
\end{equation}
The data structure additionally allows updates of the following type based on the most recently
returned triple $(G,\mathcal U,C)$:

\begin{enumerate}
    \item \textsc{ToggleCut}$(\tau)$: Given a parameter $0\le\tau\le1/\mathcal U$, the data structure implicitly adds $\tau\mathbf1_C$ to $y$.

    Then for a threshold $\xi>0$, the detector reports every edge for which $u_e(y_a-y_b)$ has changed by at least $\xi$ since insertion or last report. Over $T$ toggles, at most $\Gamma T/\xi$ reports are returned.
\end{enumerate}


The quality factor is $K=2^{O((\log m)^{5/6}\log\log m)} = m^{o(1)}$. The amortized operation factor $h$ and report factor $\Gamma$ are $m^{o(1)}$, and initialization costs $m^{1+o(1)}$.
\end{lemma}



\subsubsection{The main algorithm}

Define the fixed initial vector
\[
 p^0_s=0,
 \qquad
 p^0_v=\frac{1}{2Q}\quad(v\ne s).
\]
The cut structure maintains potentials $y$, initially zero.
We interpret the \emph{implicit} distance labels as
\begin{equation}
 p_v=p^0_v+y_v-y_s
 \qquad(v\in V).
 \label{eq:physical-labels}
\end{equation}
Thus $p_s=0$ throughout the execution.
Since $p^0$ is fixed, changes in $p_v$ are exactly changes in
$y_v-y_s$.

For each logical edge $e$, we store the capacity $\widetilde f_e$
computed at its most recent insertion or refresh.
The gradient supplied to the cut structure is
\begin{equation}
 \widetilde g_v=\sum_{\substack{e=u\to v\\e\in\mathcal E}}\widetilde f_e-\sum_{\substack{e=v\to u\\e\in\mathcal E}}\widetilde f_e,
 \qquad
 \widetilde g=-B^{\mathsf T}\widetilde f.
 \label{eq:cached-gradient}
\end{equation}
We maintain the incident sums and the gradient explicitly.
In particular, testing $\widetilde g=0$ takes constant time, since we can maintain the number of
nonzero entries.

The explicit distance array uses the scalar map
\begin{equation}
 \mathcal D(r)=
 \begin{cases}
  \infty,
    &Qr>H/2,\\[2pt]
  \displaystyle
    (1+\epsilon/2)Qr,
    &Qr\le H/2.
 \end{cases}
 \label{eq:output-rule}
\end{equation}
Whenever a return edge $r_v$ is inserted or refreshed, we immediately
overwrite $D(v)$ with $\mathcal D(p_v)$.
Set $D(s)=0$.

\paragraph{\textsc{Initialize}.}
\begin{enumerate}
 \item Set $y=0$, so that $p=p^0$.
 For every logical edge $e=a\to b$, compute
 \[
  \widetilde f_e
    =\frac{\mu_e}
      {\max\{c_e+p^0_a-p^0_b,\kappa\}}.
 \]

 \item Compute $\widetilde g$.
 Initialize the cut structure with these capacities and gradient,
 capacity bound $\kappa^{-2}$, and detector threshold $K^{-2}$, where $K\ge32$ is the approximation factor of the dynamic cut structure.

 \item Set $D(s)=0$ and
 $D(v)=\mathcal D(p^0_v)$ for every $v\ne s$.

 \item Run \textsc{Repair}.
\end{enumerate}

\paragraph{\textsc{Repair}.}
Repeat the following steps.
\begin{enumerate}
 \item Process all pending detector reports, invoking
 \textsc{Refresh}$(e)$ for every reported logical edge.

 \item If $\widetilde g=0$, stop.
 Otherwise obtain a current tuple $(G,\mathcal U,C)$ satisfying
 the dynamic cut guarantee~\eqref{eq:cut-interface}
 for capacities $\widetilde f$ and gradient $\widetilde g$.

 \item Stop if
 \[
  \frac{G}{\mathcal U}\ge-\frac{1}{4K}.
 \]

 \item Set
 \[
  \tau=\frac{1}{K^2\mathcal U}
 \]
 and toggle $C$, adding $\tau\mathbf1_C$ to the implicit
 potentials $y$.
 Return to the first step.
\end{enumerate}
We call a state \emph{stable} when \textsc{Repair} stops.
Since every stopping test follows the processing of all detector
reports, no detector report is pending at a stable state.

\paragraph{\textsc{Refresh}$(e)$.}
For a logical edge $e=a\to b$, perform the following operations.
\begin{enumerate}
 \item Query the current labels $p_a,p_b$ and compute
 \[
  f_e^{\mathrm{new}}
    =\frac{\mu_e}{\max\{c_e+p_a-p_b,\kappa\}},
  \qquad
  \Delta f_e=f_e^{\mathrm{new}}-\widetilde f_e.
 \]

 \item Delete and reinsert the physical edge corresponding to $e$
 with capacity $f_e^{\mathrm{new}}$.
 This resets its detector baseline.
 Update
 \[
 \begin{aligned}
  \widetilde g_a&\leftarrow\widetilde g_a-\Delta f_e,
  &\widetilde g_b&\leftarrow\widetilde g_b+\Delta f_e.
 \end{aligned}
 \]
 Supply the gradient changes through a paired update, and set
 $\widetilde f_e\leftarrow f_e^{\mathrm{new}}$.

 \item If $e=r_v$ is a return edge, assign
 \[
  D(v)\leftarrow\mathcal D(p_v).
 \]
\end{enumerate}

\paragraph{\textsc{Update}$(e)$.}
When an original edge $e$ is deleted, keep all labels fixed and set
\[
 c_e\leftarrow H/Q.
\]
Invoke \textsc{Refresh}$(e)$ and then \textsc{Repair}.
A permitted length increase is handled in the same way, using its
new normalized cost.

\paragraph{\textsc{Query}$(v)$.}
Return $D(v)$.




\subsection{All-target certificate}\label{sec:certificate}
In this section, we establish the following lemma:

\begin{lemma}\label{lem:certificate}
Suppose $p_s=0$, 
and the following inequality holds for every direction $x$ with $x_s=0$:
\begin{equation}\label{eq:certificate-balance}
 |g\cdot x|\leq\frac13\|x\|_f.
\end{equation}
Then every slack $\sigma_e$ exceeds $\kappa$, and simultaneously for every $v\neq s$,
\begin{equation}\label{eq:certificate-distance-bound}
 p_v\leq z_v\leq(1+\epsilon/4)p_v.
\end{equation}
\end{lemma}



\subsubsection{Backward growth finds a short-slack path}
\label{sec:backward-growth}

We first prove $z_t\le (1+\epsilon/4)p_t$ for every target $t\ne s$.
Give every core edge the positive length
\[
    \ell_e=\frac{1}{f_e}=\max\{\sigma_e,\kappa\}.
\]
For $e=u\to v$, we always have
$p_u-p_v=\sigma_e-c_e\le\ell_e$.
Fix a target $t\ne s$, and set
\[
    b_t=\max\{p_t,\kappa\},
    \qquad r_0=\frac{\epsilon}{4}b_t>0.
\]
Let $a_u^\ell=\operatorname{dist}_\ell(u,t)$, and form the backward
ball $R(r)=\{u:a_u^\ell\le r\}$.
Suppose, for a contradiction, that $s\notin R(r_0)$.
For every $r\in[0,r_0]$ and $u\in R(r)$, telescoping along a
shortest $\ell$-path from $u$ to $t$ gives
\begin{equation}
    p_u\le p_t+r
    \le \left(1+\frac{\epsilon}{4}\right)b_t,
    \qquad
    \max\{p_u,\kappa\}
    \le \left(1+\frac{\epsilon}{4}\right)b_t.
    \label{eq:backward-label-bound}
\end{equation}

Consider the continuous volume
\begin{equation}
    V_t(r)=1+
    \sum_{\substack{e=u\to v\\e\in E_0}}
    \min\left\{1,\max\left\{0,\frac{r-a_v^\ell}{\ell_e}\right\}\right\}.
    \label{eq:backward-volume}
\end{equation}
A summand with $a_v^\ell=\infty$ is defined to be zero.
Positive lengths imply $V_t(0)=1$, and always $V_t(r)\le m+1$.
Every nonzero summand has its head in $R(r)$.
Since $s\notin R(r)$, $t\in R(r)$, and $q(t)\ge1$, we have
\begin{equation}
    V_t(r)\le 1+\sum_{v\in R(r)}q(v)
    \le 2\sum_{v\in R(r)}q(v).
    \label{eq:backward-volume-mass}
\end{equation}
For almost every $r\in[0,r_0]$, each core edge $u\to v$ entering
$R(r)$ satisfies
$a_v^\ell\le r<a_u^\ell\le a_v^\ell+\ell_e$
and contributes $1/\ell_e=f_e$ to $V_t'(r)$.
Since no return edge enters $R(r)$, it follows that
$V_t'(r)\ge I(R(r))$.

Applying the cut-balance condition~\ref{eq:certificate-balance} to $\mathbf{1}_{R(r)}$ yields
$|I(R(r))-O(R(r))|\le (I(R(r))+O(R(r)))/3$, and hence
\begin{equation}
    I(R(r))\ge \frac{O(R(r))}{2}.
    \label{eq:backward-cut-balance}
\end{equation}
Every return edge from a vertex in $R(r)$ leaves the ball.
Using~\eqref{eq:backward-label-bound}, we obtain
\begin{equation}
    O(R(r))
    \ge \sum_{u\in R(r)}
        \frac{Aq(u)}{\max\{p_u,\kappa\}}
    \ge \frac{A}{(1+\epsilon/4)b_t}
        \sum_{u\in R(r)}q(u).
    \label{eq:backward-return-mass}
\end{equation}
Combining~\eqref{eq:backward-volume-mass}--\eqref{eq:backward-return-mass}
gives
\[
    \frac{V_t'(r)}{V_t(r)}
    \ge \frac{A}{4(1+\epsilon/4)b_t}
    \qquad\text{for almost every }r\in[0,r_0].
\]
Integrating contradicts the volume upper bound, since
\[
    \log V_t(r_0)
    \ge \frac{Ar_0}{4(1+\epsilon/4)b_t}
    = \frac{A\epsilon}{16(1+\epsilon/4)}
    > \log(m+1).
\]
The strict inequality follows from
$A\epsilon\ge32\log_2(m+1)$ and $\epsilon\le1$.

There is therefore an $s$--$t$ core path $\pi$ of total $\ell$-length
at most $r_0=(\epsilon/4)b_t$.
Since $p_s=0$, its cost satisfies
\[
    z_t\le \sum_{e\in\pi}c_e
    =p_t+\sum_{e\in\pi}\sigma_e
    \le p_t+\sum_{e\in\pi}\ell_e
    \le p_t+\frac{\epsilon}{4}b_t.
\]
Every source-to-target core path has cost at least $1/Q$.
If $p_t\le\kappa$, then $b_t=\kappa$, so the preceding inequality
would imply
\[
    \frac{1}{Q}\le z_t
    \le \left(1+\frac{\epsilon}{4}\right)\kappa
    <\frac{1}{Q},
\]
a contradiction.
Here the last inequality follows from
$\kappa=(AmQ)^{-2}$ and the parameter choices.
Consequently $p_t>\kappa$, so $b_t=p_t$ and
$z_t\le(1+\epsilon/4)p_t$.
Since $t$ was arbitrary, we conclude that for every $v\ne s$,
\begin{equation}
    p_v\ge\frac{1}{(1+\epsilon/4)Q}>\kappa>0.
    \label{eq:backward-positive-labels}
\end{equation}

\subsubsection{A successor cut restores strict feasibility}\label{subsec:certificate-feasibility}

Now we prove that every slack $\sigma_e$ exceeds $\kappa$.
Let $D_{\mathrm{high}}$ consist of core edges
with $f_e\geq2Q$.  Every such edge $e=u\to v$ satisfies
\[
 p_v-p_u=c_e-\sigma_e
 \geq\frac1Q-\frac1{f_e}
 \geq\frac1{2Q}>0.
\]
Thus $D_{\mathrm{high}}$ is acyclic.

Fix a high-capacity edge $u\to v$, and let $R$ contain $v$ and all its
successors in $D_{\mathrm{high}}$.  Acyclicity gives $u\notin R$.
Moreover $s\notin R$, since labels strictly increase along high-capacity
edges, all nonsource labels are positive, and $p_s=0$.  No high-capacity
edge leaves $R$, so 
$$
O(R)<2mQ + \sum_{v\in R}f_{v\to s}\le 2mQ+2AmQ.
$$
The upper cut inequality gives
\[
 f_{uv}\leq I(R)
 \leq2\cdot O(R)
 \leq4mQ+4AmQ\leq10AmQ.
\]
Edges outside $D_{\mathrm{high}}$ obey the same upper bound.  Since
\[
 10AmQ<\frac1\kappa=(AmQ)^2,
\]
no core constraint is clipped.

\begin{proof}[Proof of Lemma~\ref{lem:certificate}]
The only remaining part is to prove $p_v\le z_v$.
Since every slack $\sigma_e$ exceeds $\kappa$, for every core edge $e=u\to v$,
$$
    p_v-p_u=c_e-\sigma_e\le c_e.
$$
Summing these inequalities along a shortest core path from $s$ to
$v$ gives $p_v\leq z_v$.
\end{proof}

\subsection{Analysis of the main algorithm}
\label{sec:algorithm-analysis}


\subsubsection{Cache accuracy}
\label{subsec:analysis-cache}

\begin{lemma}
\label{lem:cache}
After all pending detector reports have been processed, the following inequalities hold:
\begin{equation}
 \label{eq:cache-error-1}
 (1-K^{-2})f_e\le\widetilde f_e\le(1+K^{-2})f_e
\qquad(e\in\mathcal E).
\end{equation}
\begin{equation}
 \label{eq:cache-error-2}
 |(g-\widetilde g)\cdot x|\le K^{-2}\|x\|_f
       \qquad(x\in\mathbb R^V).
\end{equation}
\begin{equation}
 \label{eq:cache-error-3}
 (1-K^{-2})\|x\|_f\le\|x\|_{\widetilde f}
       \le(1+K^{-2})\|x\|_f\qquad(x\in\mathbb R^V).
\end{equation}
\end{lemma}
\begin{proof}
Fix an edge $e=u\to v$ that has not been refreshed at the current labels. Write
$\sigma_e^{\mathrm{old}}$ for its slack at its most recent refresh and
put $t_e^{\mathrm{old}}=\max\{\sigma_e^{\mathrm{old}},\kappa\}$.
Its cost has not changed since that refresh, because a cost increase
forces another refresh, and its coefficient is fixed.  Since $p^0$ is
fixed, a change in $p_u-p_v$ is exactly the change in $y_u-y_v$.
The absence of a pending detector report therefore gives
\[
 |\sigma_e-\sigma_e^{\mathrm{old}}|
 <\frac{1}{K^2\widetilde f_e}
 =\frac{t_e^{\mathrm{old}}}{K^2\mu_e}.
\]
The clipping map $t\mapsto\max\{t,\kappa\}$ is $1$-Lipschitz, so
\[
 \left|\frac{\widetilde f_e}{f_e}-1\right|
 =\left|\frac{t_e}{t_e^{\mathrm{old}}}-1\right|
 \le\frac{1}{K^2\mu_e}\le K^{-2}.
\]
Using
$g-\widetilde g=-B^{\mathsf T}(f-\widetilde f)$ now yields
\[
 |(g-\widetilde g)\cdot x|
 \le\sum_e|f_e-\widetilde f_e|\,|(Bx)_e|
 \le K^{-2}\|x\|_f.
\]
\end{proof}

\subsubsection{Descent and preservation of the energy bound}
\label{subsec:analysis-energy}
In this section, we give lower and upper bounds on the initial value of $\Psi$ and prove a lower bound on its decrease in each repair iteration.

At initialization every return slack equals $1/(2Q)$.  A core edge
between two nonsource vertices has slack equal to its cost, at least
$1/Q$.  A core edge leaving $s$ has slack at least $1/(2Q)$.  Thus every initial slack is at least
$1/(2Q)>\kappa$.  Since the total
coefficient is at most $m+Am$, we obtain
\begin{equation}
 0\le\Psi_{\mathrm{init}}
 \le(m+Am)\log(2Q)<2Am\log(2Q).
 \label{eq:initial-energy}
\end{equation}

Every external update increases one core cost while holding the labels
fixed.  The clipped logarithm is nonincreasing, so this cannot increase
$\Psi$.  Its return terms and all coefficients remain unchanged.
Refreshing the affected edge, or any reported edge, has no effect on
$\Psi$ because it changes neither the costs nor the physical labels.

\begin{lemma}
\label{lem:descent}
Each toggle decreases the potential by at least
$1/(8K^3)$.
\end{lemma}
\begin{proof}
Let $x$ be the physical displacement produced by the toggle.  Because
$p^0$ is fixed,
\[
x_v=\Delta p_v
=\Delta y_v-\Delta y_s
=\tau\bigl(\mathbf1_C(v)-\mathbf1_C(s)\bigr),
 \qquad x_s=0.
\]
The cut-capacity upper bound and
Lemma~\ref{lem:cache} imply
\begin{equation}
 \|x\|_f
 \le\frac{\tau\,\operatorname{cut}_{\widetilde f}(C)}{1-K^{-2}}
 \le\frac{\tau\mathcal U}{1-K^{-2}}
 =\frac1{K^2-1}.
 \label{eq:step-norm-bound}
\end{equation}
Recall that
\[
G=\widetilde g(C)=\sum_{v\in C}\widetilde g_v.
\]
Thus $\widetilde g\cdot x=\tau G$.  Since the stopping test failed, Lemma~\ref{lem:cache} gives
\begin{equation}
 \begin{split}
 g\cdot x
    &=\widetilde g\cdot x+(g-\widetilde g)\cdot x\\
 &\le\tau G+K^{-2}\|x\|_f\\
 &\le-\frac1{4K^3}+\frac2{K^4}
 =-\frac{1/4-2/K}{K^3}.
 \end{split}
 \label{eq:linear-descent}
\end{equation}

We bound the Taylor remainder without assuming that the intermediate
slacks are positive.  All capacities and denominators in the following
calculation are evaluated at the start of the toggle.  Put
$d_e=(Bx)_e$, and define
\[
 \Phi(\lambda)=\Psi(p+\lambda x),
 \qquad
 t_e(\lambda)=\max\{\sigma_e+\lambda d_e,\kappa\}
 \qquad(0\le\lambda\le1).
\]
Let $r_{\mathrm{step}}:= \frac1{K^2-1}$ be an upper bound on $\|x\|_f$. Then
\[
 \frac{|d_e|}{t_e}
 =\frac{f_e|d_e|}{\mu_e}
 \le\frac{r_{\mathrm{step}}}{\mu_e}
 \le r_{\mathrm{step}}.
\]
The Lipschitz property of clipping consequently yields
\[
 t_e(\lambda)\ge t_e-|d_e|
       \ge(1-r_{\mathrm{step}})t_e.
\]
The derivative of $\psi_\kappa$ is Lipschitz; hence
$\Phi'$ is absolutely continuous.  Almost everywhere along the
segment, $\psi_\kappa''$ is zero below the clipping threshold and
$1/\sigma^2$ above it.  The preceding denominator bound therefore gives
\[
 \Phi''(\lambda)
 \le\frac1{(1-r_{\mathrm{step}})^2}
       \sum_e\frac{\mu_e d_e^2}{t_e^2}
 =\frac1{(1-r_{\mathrm{step}})^2}
       \sum_e\frac{(f_e d_e)^2}{\mu_e}
 \quad\text{for almost every }\lambda.
\]
Absolute continuity justifies the integral Taylor identity, including
when an edge crosses the clipping threshold.  As $\mu_e\ge1$,
\begin{align*}
 \Phi(1)-\Phi(0)-g\cdot x
 &=\int_0^1(1-\lambda)\Phi''(\lambda)\,d\lambda\\
 &\le\frac1{2(1-r_{\mathrm{step}})^2}
       \sum_e\frac{(f_e d_e)^2}{\mu_e}\\
 &\le\frac{\bigl(\sum_e f_e|d_e|\bigr)^2}
              {2(1-r_{\mathrm{step}})^2}\\
 &\le\frac{r_{\mathrm{step}}^2}
              {2(1-r_{\mathrm{step}})^2}
 \le\frac1{K^4}.
\end{align*}
Combining this estimate with~\eqref{eq:linear-descent} proves
\[
 \Psi(p)-\Psi(p+x)
 \ge\frac{1/4-3/K}{K^3}
 >\frac1{8K^3}.
\]
\end{proof}

\begin{lemma}
\label{lem:potential-nonnegative}
If $p_s=0$ holds, then $\Psi(p)\ge0$.
\end{lemma}
\begin{proof}
Since $\kappa<1$, we have $\psi_\kappa(t)\ge0$ whenever $t\le1$.
Thus, if $L:=\max_{e\in\mathcal E}\sigma_e\le1$, every term of $\Psi$
is nonnegative and the claim follows.

Suppose now that $L>1$, and choose an edge $e$ with $\sigma_e=L$.
Every logical edge lies on a closed directed walk with at most
$n+1$ edges and total cost at most $1/8$.
Indeed, for a core edge $u\to v$, take a fixed simple core path
from $s$ to $u$, followed by $u\to v$ and the return edge $v\to s$.
For a return edge $v\to s$, take a fixed simple core path from
$s$ to $v$, followed by that return edge.
Each walk contains at most $n$ core edges, each of cost at most
$1/(8n)$.

Telescoping the slacks around such a walk shows
that the other edges have total slack
at most $1/8-L<-L/2$.
Hence some edge $a$ satisfies $\sigma_a<-L/(2n)$, and therefore
\[
    \mu_a\psi_\kappa(\sigma_a)
    =\mu_a\left(-\log\kappa+1-\frac{\sigma_a}{\kappa}\right)
    >\frac{L}{2n\kappa}.
\]
On the other hand, $\psi_\kappa$ is decreasing and every slack
is at most $L$, so for any $b\in\mathcal E$,
\[
    \psi_\kappa(\sigma_b)\ge\psi_\kappa(L)=-\log L.
\]
Using $\sum_{b\in\mathcal E}\mu_b\le2Am$ and $\log L\le L$, we obtain
\[
    \Psi(p)
    >\frac{L}{2n\kappa}-2Am\log L
    \ge L\left(\frac{1}{2n\kappa}-2Am\right)
    >0.
\]
The last inequality follows from
$\kappa=(AmQ)^{-2}$ and $Q=32n^2W$.
\end{proof}

For completeness, we prove that all numerical quantities required by the cut data structure have
polynomial bounds. Indeed, since $W=m^{O(1)}$ and
$\epsilon^{-1}\le m$, the parameters $A,Q,\kappa^{-1}$ and
$\mu_e$ are polynomially bounded. Throughout the execution,
$\Psi(p)\le 2Am\log(2Q)$, so the estimate in the preceding proof
gives $\max_e \sigma_e\le m^{O(1)}$. The return slacks and
telescoping along simple core paths from $s$ then imply
$|p_v|\le m^{O(1)}$. Consequently,
\[
    m^{-O(1)}\le f_e,\widetilde f_e\le m^{O(1)},
    \qquad
    |\widetilde g_v|\le m^{O(1)}.
\]


\subsubsection{The stopping condition certifies all labels}
\label{subsec:analysis-stopping}

\begin{lemma}[Stable labels]
\label{lem:stable}
At every stable state,
\[
 |g\cdot x|\le\frac13\|x\|_f
 \qquad\text{for all }x\in\mathbb R^V\text{ with }x_s=0.
\]
Consequently the all-target certificate of
Lemma~\ref{lem:certificate} applies.
\end{lemma}
\begin{proof}
No report is pending when the algorithm stops, so
Lemma~\ref{lem:cache} applies.  If $\widetilde g=0$, it gives
$|g\cdot x|\le K^{-2}\|x\|_f\le\|x\|_f/3$ immediately.

Otherwise, let $r^*$ be the minimum ratio for the current cached
capacities $\widetilde f$ and gradient $\widetilde g$.  The cut
guarantee~\eqref{eq:cut-interface} and the stopping test imply
\[
 r^*\ge K\frac G{\mathcal U}\ge-\frac14.
\]
Applying this bound to $x$ and $-x$ yields
$|\widetilde g\cdot x|\le\|x\|_{\widetilde f}/4$.
The assertion is also valid for a zero-norm direction, since the
physical graph is connected and the cached gradient has zero sum.
The cache bounds now imply
\[
 |g\cdot x|
 \le\frac14\|x\|_{\widetilde f}+K^{-2}\|x\|_f
 \le\left(\frac{1+K^{-2}}{4}+K^{-2}\right)\|x\|_f
 \le\frac13\|x\|_f,
\]
where $K\ge32$ ensures $K^{-2}\le1/100$.
\end{proof}

\subsubsection{The accuracy of the most recent return refresh}
\label{subsec:return-snapshot}

For analysis only, let $\bar p_v$ denote the value of $p_v$ at the most recent
insertion or refresh of the return constraint $v\to s$.  The algorithm does
not need a separate field storing $\bar p_v$.  At that refresh the exact
stored capacity of its physical edge was
\[
    \widetilde f_v=\frac{Aq(v)}{\max\{\bar p_v,\kappa\}}.
\]
Recall the parameter choices from Subsection~\ref{subsec:main-algorithm}.  They give
\begin{equation}
    \frac{2}{AK^2}\le\frac{\epsilon}{8}.
    \label{eq:snapshot-rho}
\end{equation}

\begin{lemma}[Return-refresh accuracy]
\label{lem:return-snapshot}
At every stable state and every $v\ne s$,
\begin{equation}
    (1-\epsilon/8)p_v\le\bar p_v\le(1+\epsilon/8)p_v.
    \label{eq:return-snapshot}
\end{equation}
\end{lemma}

\begin{proof}
All detector reports are processed before the stopping test is applied.
Unless the return was just refreshed, its unreported change therefore
satisfies
\[
    \omega_v:=|p_v-\bar p_v|
       \le\frac{1}{K^2\widetilde f_v}
       =\frac{1}{AK^2q(v)}\max\{\bar p_v,\kappa\}.
\]

At the current stable state, Lemma~\ref{lem:certificate} gives
$p_v>\kappa$.  Hence
$\max\{\bar p_v,\kappa\}\le p_v+\omega_v$.
Since $K^{-2}/(Aq(v))<1/2$,
\[
    \omega_v\le\frac{1}{AK^2q(v)}(p_v+\omega_v)
      \quad\Longrightarrow\quad
    \omega_v\le\frac{2}{AK^2q(v)}p_v\le\frac{\epsilon}{8}p_v.
\]
For a just-refreshed return, $\omega_v=0$.
Thus we obtain \eqref{eq:return-snapshot}.
\end{proof}

For every $v\ne s$, the value of $D(v)$ is $\mathcal D(\bar p_v)$; recall that $D(s)=0$.
From \eqref{eq:certificate-distance-bound} and \eqref{eq:return-snapshot}, we conclude that 
for any reachable $v$,
\begin{equation}\label{eq:decremental}
d_i(v)\le D(v)\le (1+\epsilon)d_i(v).
\end{equation}
For any unreachable $v$, \eqref{eq:distance-gap}, \eqref{eq:certificate-distance-bound}, and \eqref{eq:return-snapshot} give
\[
Q\bar p_v
\ge
\frac{1-\epsilon/8}{1+\epsilon/4}\,Qz_v
\ge \frac7{10}H>H/2.
\]
Thus $D(v) = \infty$.
\subsection{Total running time}
\label{sec:running-time}




Initially all constraints are strictly feasible and unclipped.  The
initial-energy bound~\eqref{eq:initial-energy} gives
\begin{equation}
    \Psi_{\mathrm{init}}\le2Am\log(2Q).
    \label{eq:runtime-initial-budget}
\end{equation}
Every external update increases the cost of
one core constraint.  
Since the clipped logarithm $\psi_\kappa$ is non-increasing, this does not increase the potential.

Lemma~\ref{lem:descent} gives a decrease of at least $1/(8K^3)$ for each
exact repair toggle.  Lemma~\ref{lem:potential-nonnegative} gives $\Psi\ge0$.  
Consequently the total number $T$ of repair toggles,
including initialization repair, satisfies
\begin{equation}
    T=O(K^3Am\log(2Q)) = m^{1+o(1)}\epsilon^{-1}.
    \label{eq:runtime-toggles}
\end{equation}



One logical refresh queries its two endpoints, replaces its physical edge,
and changes its two endpoint gradient entries through paired updates.
These are $O(1)$ interface calls.  If it is a return refresh, it also invokes
the map $\mathcal D$ from~\eqref{eq:output-rule}, at constant additional
cost.

After each toggle, process all reports and refresh each reported logical
constraint once.  All
other work of a toggle, including obtaining a current cut tuple if needed,
takes $O(1)$ interface calls.  Let $R$ be the total number of reports.  Then
\Cref{prop:cut-interface} gives
\begin{equation}
    R\le\Gamma K^2T = m^{1+o(1)}{\epsilon}^{-1}.
    \label{eq:runtime-reports}
\end{equation}
Note that the initialization takes $m^{1+o(1)}$ time, each interface call takes $m^{o(1)}$ time,
and each update to a distance estimate takes $O(1)$ time. Thus the total time complexity is bounded by 
$$
{m^{1+o(1)}{\epsilon}^{-1}}.
$$
\begin{proof}[Proof of \Cref{thm:decremental}]

\eqref{eq:decremental} establishes that, when the state is stable, for any $v\in V$,
$$d_i(v)\le D(v)\le (1+\epsilon)d_i(v).$$
Combining this guarantee with the running-time bound of ${m^{1+o(1)}{\epsilon}^{-1}}$ proves \Cref{thm:decremental}.

\end{proof}

\section{Incremental SSSP}\label{sec4:incremental}

The goal of this section is to establish the following lemma.

\begin{lemma}\label{sec4:main}
Let $m\ge n\ge2$, $W=m^{O(1)}$, and $1/m<\epsilon<1$.
Let $G$ be an $n$-vertex directed graph with a fixed source $s$ and
positive integer edge weights in $[1,W]$, undergoing edge insertions.
Assume that $G$ is initially empty, and the number of insertions is at most $m$.
There is a deterministic data structure in the unit-cost arithmetic
model that explicitly maintains $D:V\to\mathbb R_{\ge0}\cup\{\infty\}$
such that
\[
 \operatorname{dist}_G(s,v)\le D(v)
 \le(1+\epsilon)\operatorname{dist}_G(s,v)
 \qquad(v\in V),
\]
in total update time $m^{1+o(1)}\epsilon^{-1}$.
\end{lemma}

We define the parameters in \Cref{sec4:notation}, describe the dynamic
cycle interface and the algorithm in \Cref{sec4:algorithm}, prove the
potential and distance guarantees in \Cref{sec4:analysis}, and bound
the total running time in \Cref{sec4:time}.

\subsection{Basic Notation}\label{sec4:notation}

Let $G_i=(V,E_i,w)$ be the original graph after $i$ insertions.
We discard every initial or inserted edge entering $s$; such an edge
does not occur on a simple shortest path from $s$.
Set
\begin{equation}\label{sec4:normalization}
 H=4nW,\qquad Q=8nH.
\end{equation}
Add a sentinel edge $a_v=s\to v$ of weight $H$ for every $v\ne s$,
and divide all weights by $Q$.
Write $\widehat G_i=(V,\widehat E_i,c)$ for the resulting core graph.
There are at most $m+n-1\le2m$ core edges over the entire sequence,
and every core edge satisfies
\begin{equation}\label{sec4:core-costs}
 \frac1Q\le c_e\le\frac1{8n}.
\end{equation}
For $v\ne s$, write $d_i(v)=\operatorname{dist}_{G_i}(s,v)$ and
$z_i(v)=\operatorname{dist}_{\widehat G_i}(s,v)$. Then
\begin{equation}\label{sec4:sentinel-distances}
 Qz_i(v)=
 \begin{cases}
  d_i(v)<H/4,&d_i(v)<\infty,\\
  H,&d_i(v)=\infty.
 \end{cases}
\end{equation}
Indeed, a finite original distance is at most $(n-1)W$, whereas a
path using a sentinel has weight at least $H$; a direct sentinel is
always available. In particular,
$1/Q\le z_i(v)\le1/(8n)<1/8$.

For each core edge $e=u\to v$, add a distinct ordinary return
$r_e=v\to s$ of cost zero. The base edge set is
$\mathcal B_i=\widehat E_i\sqcup\{r_e:e\in\widehat E_i\}$.
If $q_i(v)$ denotes the core indegree of $v$, then $v$ has exactly
$q_i(v)$ ordinary returns, and
\[
 q_i(v)\ge1\quad(v\ne s),\qquad
 \sum_{v\ne s}q_i(v)=|\widehat E_i|\le2m.
\]

Choose one common bound
\begin{equation}\label{sec4:overhead}
 K=\exp\bigl(O((\log m)^{167/168}\log\log m)\bigr)
   =m^{o(1)}\ge2^{20},
\end{equation}
and set
\begin{equation}\label{sec4:parameters}
 \begin{gathered}
 A=\min\{4^j:j\in\mathbb Z_{\ge0},\
       4^j\ge2^{12}\lceil\log_2m\rceil/\epsilon\},\\
 M=mAQK,\qquad\kappa=M^{-20},\qquad\tau=\epsilon/64.
 \end{gathered}
\end{equation}
Thus $A=\Theta(\epsilon^{-1}\log m)$ and $\sqrt A$ is an integer.
Each core edge has coefficient $\mu_e=1$, and each ordinary return
has coefficient $\mu_e=A$.

For each $v\ne s$, also add a probe $\pi_v=v\to s$.
It has coefficient $\mu_{\pi_v}=1$, threshold $\theta_v$, upper
flow parameter $u_v=(\tau\theta_v)^{-1}$, and mathematical cost
$c_{\pi_v}=-\theta_v$.
Initially $\theta_v=1/4$ and every probe has zero flow.
The logical edge set is
$\mathcal E_i=\mathcal B_i\sqcup\{\pi_v:v\ne s\}$.

Define the clipped logarithm
\begin{equation}\label{sec4:clipped-log}
 \psi_\kappa(t)=
 \begin{cases}
  -\log t,&t\ge\kappa,\\
  -\log\kappa+1-t/\kappa,&t<\kappa.
 \end{cases}
\end{equation}
It is convex and continuously differentiable, with
$\psi_\kappa'(t)=-1/\max\{t,\kappa\}$.
Let $B$ be the tail-minus-head incidence matrix of the logical graph.
For a signed circulation $x$, meaning $B^\top x=0$, define
\begin{equation}\label{sec4:potential}
 \Phi(x)=c^\top x+
 \sum_{e\in\mathcal E_i}\mu_e\psi_\kappa(x_e+\kappa)
 +\sum_{v\ne s}
   \bigl(\psi_\kappa(u_v-x_{\pi_v})+\log u_v\bigr).
\end{equation}

For every logical edge put $t_e=\max\{x_e+\kappa,\kappa\}$.
The gradient $g=\nabla\Phi(x)$ and positive metric $h$ are
\begin{equation}\label{sec4:edge-coefficients}
 \begin{aligned}
 g_e&=c_e-\mu_e/t_e,&
 h_e&=\sqrt{\mu_e}/t_e
 &&(e\in\mathcal B_i),\\
 t_v^+&=\max\{u_v-x_{\pi_v},\kappa\},&&&\\
 g_{\pi_v}&=-\theta_v-1/t_{\pi_v}+1/t_v^+,&
 h_{\pi_v}&=1/t_{\pi_v}+1/t_v^+
 &&(v\ne s).
 \end{aligned}
\end{equation}

\subsection{Algorithm}\label{sec4:algorithm}

\subsubsection{The dynamic cycle interface}

For positive edge weights $u$ and an edge vector $z$, write
$\|uz\|_1=\sum_eu_e|z_e|$.
We use the following interface on graphs with $O(m)$ logical edges,
polynomially bounded gradients, and positive metrics with polynomial
upper and inverse-polynomial lower bounds.

\begin{lemma}[Adapted from {\cite[Theorem~3.8]{CKLMP24}}]\label{prop:cycle-interface}
Fix $r,\Gamma,\xi>0$.
There is a deterministic data structure that implicitly maintains
a signed circulation $x$, initially zero, on a graph with $O(m)$
logical edges. Write $u>0$ and $\gamma$ for the currently supplied
metric and gradient. It supports the following operations.
\begin{enumerate}
\item \textsc{InsertEdge}$(e,u_e,\gamma_e)$ inserts an edge with
the specified metric and gradient, initializing its flow to zero.
\item \textsc{ApplyCycle}$()$ either adds a circulation $\Delta$ satisfying
\[
 \gamma^\top\Delta=-\Gamma,\qquad
 \|u\Delta\|_1\le\Gamma/r,
\]
and returns the reported edges with their exact current flows,
or certifies
\[
 |\gamma^\top z|\le Kr\|uz\|_1
 \qquad(B^\top z=0).
\]
\item \textsc{QueryFlow}$(e)$ returns the exact current flow of any
current edge.
\item \textsc{UpdateEdge}$(e,u'_e,\gamma'_e)$ updates the metric and gradient of $e$.
\end{enumerate}
The algorithm maintains a queue of pending notifications, each
containing a reported edge and its exact current flow.
Successful \textsc{ApplyCycle} calls and \textsc{UpdateEdge} calls
may generate notifications. 
Let $x_e^{\rm last}$ be the exact flow at initialization, insertion,
or the most recent \textsc{UpdateEdge} call.
After all pending notifications and coefficient updates have been
processed,
\[
 u_e|x_e-x_e^{\rm last}|\le\xi\qquad(e\in E).
\]
Each ordinary notification is processed by one \textsc{UpdateEdge}
call or cleared as a duplicate by an update of the same edge.
An \textsc{UpdateEdge} call clears old pending notifications for its
edge and retains notifications for other edges.
All pending notifications are processed before the next
\textsc{ApplyCycle} call.

Assume that the supplied gradients and metrics remain within the bounds specified above.
Over $L$ cycle calls, of which $T$ are successful, let $R$ count
\textsc{QueryFlow} calls and let $F$ count \textsc{UpdateEdge} calls
not used to process ordinary notifications, including calls that
retain the same coefficients. The total maintenance time is
\[
 O\!\left(
   \left(m+L+F+R+\frac{T\Gamma}{\xi r}\right)m^{o(1)}
 \right).
\]
\end{lemma}

\begin{proof}
Use the construction of
\cite[Theorems~3.5--3.8 and Lemma~2.7]{CKLMP24},
retaining conservative expanded-path lengths $\Lambda(d)$.
An accepted direction satisfies
$-\gamma^\top d\ge r\Lambda(d)$, so scaling the direction $d$ to obtain
$\gamma^\top d=-\Gamma$ gives represented movement at most $\Gamma/r$.
Thus the total represented weighted variation is
$\mathcal W\le T\Gamma/r$.
Read-only ratio queries and symmetry under $z\mapsto-z$
give the stopping certificate.

Let $h=m^{o(1)}$ bound the forest multiplicity.
For each logical edge, maintain materialized absolute variation
with threshold $\xi/2$ and residual variation on its forest copies
with threshold $\xi/(2h)$.
Representation changes transfer residual variation without
duplication or baseline resets.
Hence the unreported error is at most $\xi$, and each ordinary
notification consumes at least $\xi/(2h)$ weighted variation.
There are therefore $O(h\mathcal W/\xi)$ such notifications.

An \textsc{UpdateEdge} call resets the baseline to the exact current
flow and clears the old-metric counters before installing new
coefficients.
An edge-indexed queue clears only the affected edge's duplicates.
A \textsc{QueryFlow} call uses inverse lists and dynamic-tree point
queries to return the explicit flow plus the signed implicit flows,
without changing detection records or the queue.
Each basic operation and processed notification costs $m^{o(1)}$.
Summing over the $L$ cycle calls, $F$ additional updates, $R$ queries,
and ordinary notifications proves the claimed bound.
\end{proof}


\subsubsection{The main algorithm}

For each edge, store an exact flow sample $\bar x_e$, initialized
to zero at insertion. Update this sample during ordinary refreshes
and threshold transactions as described below. On base edges set
\begin{equation}\label{sec4:cached-coefficients}
 \bar t_e=\max\{\bar x_e+\kappa,\kappa\},\qquad
 \bar g_e=c_e-\mu_e/\bar t_e,\qquad
 \bar h_e=\sqrt{\mu_e}/\bar t_e.
\end{equation}
For each probe $\pi_v$, set
\[
 \begin{aligned}
 \bar t_{\pi_v}
 &=\max\{\bar x_{\pi_v}+\kappa,\kappa\},
 &
 \bar t_v^+
 &=\max\{u_v-\bar x_{\pi_v},\kappa\},\\
 \bar g_{\pi_v}
 &=-\theta_v-\frac1{\bar t_{\pi_v}}
                +\frac1{\bar t_v^+},
 &
 \bar h_{\pi_v}
 &=\frac1{\bar t_{\pi_v}}+\frac1{\bar t_v^+}.
 \end{aligned}
\]
Maintain a list of flagged probes:
$v$ is flagged exactly when $\bar x_{\pi_v}\ge u_v/2$.
Initialization and every threshold transaction write
\begin{equation}\label{sec4:readout-definition}
 E(v)=(1+\epsilon/4)Q\theta_v\quad(v\ne s),\qquad E(s)=0,
\end{equation}
and
\begin{equation}\label{sec4:public-definition}
 D(v)=
 \begin{cases}
  E(v),&E(v)\le H/2,\\
  \infty,&E(v)>H/2,
 \end{cases}
 \qquad D(s)=0.
\end{equation}

\paragraph{Initialize.}
Construct the logical graph, set $x=\bar x=0$, and initialize all
cached coefficients, flags, and public entries.
Initialize the cycle implementation with
\begin{equation}\label{sec4:solver-parameters}
 r=K^{-2},\qquad\Gamma=K^{-6},\qquad\xi=K^{-4},
 \qquad\gamma=\bar g,\qquad u=\bar h,
\end{equation}
and run \textsc{Repair}.

\paragraph{Repair.}
Repeat the following steps.
\begin{enumerate}
\item Process all pending notifications, invoking \textsc{Refresh}
on every reported logical edge, until the notification queue is empty.
\item Call \textsc{ApplyCycle}. On a successful augmentation,
\begin{equation}\label{sec4:augmentation}
 \bar g^\top\Delta=-K^{-6},\qquad
 \|\bar h\Delta\|_1\le K^{-4}.
\end{equation}
Return to step 1.
\item Otherwise the current, unchanged coefficients satisfy
\begin{equation}\label{sec4:stopping-certificate}
 |\bar g^\top z|\le K^{-1}\|\bar hz\|_1
 \qquad(B^\top z=0).
\end{equation}
There are no pending notifications. Call this state
\emph{cycle-stationary}. If no probe is flagged, return.
Otherwise select a flagged vertex $v$ and perform the following
threshold transaction.

The flag in this certified state authorizes the transaction.
Set $\bar x_{\pi_v}\leftarrow\textsc{QueryFlow}(\pi_v)$.
Then set
\[
 \theta_v\leftarrow(1-\tau)\theta_v,\qquad
 u_v\leftarrow u_v/(1-\tau),\qquad
 c_{\pi_v}\leftarrow-\theta_v.
\]
Recompute the cached coefficients and call
$\textsc{UpdateEdge}(\pi_v,\bar h_{\pi_v},\bar g_{\pi_v})$
to supply the new gradient and metric. Commit one flag update for $v$
using the new parameters, and write
\eqref{sec4:readout-definition}--\eqref{sec4:public-definition}.
Notifications for other edges remain pending. Return to step 1.
\end{enumerate}
A state is \emph{stable} when \textsc{Repair} returns.
Public distance queries are served only in stable states.

\paragraph{Refresh$(e)$.}
Replace $\bar x_e$ by the exact flow contained in its notification,
recompute its cached coefficients, and call
$\textsc{UpdateEdge}(e,\bar h_e,\bar g_e)$.
This resets its detection baseline and clears duplicate notifications
for $e$; notifications for other edges are retained.
For a probe, also update its flag.
An ordinary refresh does not write a distance entry.

\paragraph{Insert$(e)$.}
If the inserted edge enters $s$, discard it. Otherwise insert the
normalized core edge and its distinct ordinary return at zero flow,
initialize their samples and coefficients, and run \textsc{Repair}.

\paragraph{Query$(v)$.}
Return the explicit entry $D(v)$.

\subsection{Analysis of the Main Algorithm}\label{sec4:analysis}

We first establish cache accuracy and derive the distance
certificates at cycle-stationary states.
We then prove potential descent and the threshold invariant, which
we use to establish the global budget and numerical ranges.
We conclude with the accuracy of the explicit distance readout.

\subsubsection{Cache accuracy}

\begin{lemma}[Cache accuracy]\label{sec4:cache}
Let $\nu=K^{-4}$. After all pending reports have been processed
and every parameter transaction has completed,
\begin{equation}\label{sec4:cache-widths}
 \begin{gathered}
 \bar h_e|x_e-\bar x_e|\le\nu\qquad(e\in\mathcal E_i),\\
 \left|t_e/\bar t_e-1\right|\le\nu/\sqrt{\mu_e}
 \qquad(e\in\mathcal B_i),\\
 \max\left\{
  \left|t_{\pi_v}/\bar t_{\pi_v}-1\right|,
  \left|t_v^+/\bar t_v^+-1\right|
 \right\}\le\nu\qquad(v\ne s).
 \end{gathered}
\end{equation}
On all logical edges, including probes,
\begin{equation}\label{sec4:cache-comparisons}
 (1-\nu)h_e\le\bar h_e\le(1+\nu)h_e,
 \qquad |g_e-\bar g_e|\le\nu h_e\le2\nu\bar h_e.
\end{equation}
\end{lemma}

\begin{proof}
The first inequality is the detector invariant, since
$\bar x_e=x_e^{\rm last}$ after all pending notifications have been
processed and all parameter transactions have completed. Write
$d_e=|x_e-\bar x_e|$.
On a base edge,
$d_e\le\nu/\bar h_e=\nu\bar t_e/\sqrt{\mu_e}$.
Since clipping is $1$-Lipschitz, this gives $\left|t_e/\bar t_e-1\right|\le\nu/\sqrt{\mu_e}$,
and
\[
 |g_e-\bar g_e|
 \le\frac{\mu_ed_e}{t_e\bar t_e}
 \le\nu\frac{\sqrt{\mu_e}}{t_e}=\nu h_e.
\]
For a probe,
\[
 d_{\pi_v}\le\frac\nu{\bar h_{\pi_v}}
 \le\nu\min\{\bar t_{\pi_v},\bar t_v^+\}.
\]
The same Lipschitz argument gives
\[
 \begin{aligned}
 |g_{\pi_v}-\bar g_{\pi_v}|
 &\le d_{\pi_v}
 \left(\frac1{t_{\pi_v}\bar t_{\pi_v}}
             +\frac1{t_v^+\bar t_v^+}\right)\\
 &\le\nu\left(\frac1{t_{\pi_v}}+\frac1{t_v^+}\right)
 =\nu h_{\pi_v}.
 \end{aligned}
\]
For $e\in\mathcal B_i$, the definitions give
\[
 \frac{\bar h_e}{h_e}
 =\frac{\sqrt{\mu_e}/\bar t_e}{\sqrt{\mu_e}/t_e}
 =\frac{t_e}{\bar t_e}.
\]
Since $\mu_e\ge1$,
\[
 1-\nu
 \le1-\frac{\nu}{\sqrt{\mu_e}}
 \le\frac{\bar h_e}{h_e}
 \le1+\frac{\nu}{\sqrt{\mu_e}}
 \le1+\nu.
\]
Multiplying by $h_e>0$ yields
\[
 (1-\nu)h_e\le\bar h_e\le(1+\nu)h_e
 \qquad(e\in\mathcal B_i).
\]

For a probe $\pi_v$, the Lipschitz property of clipping and
the bound on $d_{\pi_v}$ give
\[
 |t_{\pi_v}-\bar t_{\pi_v}|
 \le d_{\pi_v}\le\nu\bar t_{\pi_v},
 \qquad
 |t_v^+-\bar t_v^+|
 \le d_{\pi_v}\le\nu\bar t_v^+.
\]
Hence
\[
 1-\nu\le\frac{t_{\pi_v}}{\bar t_{\pi_v}}\le1+\nu,
 \qquad
 1-\nu\le\frac{t_v^+}{\bar t_v^+}\le1+\nu.
\]
Dividing the first comparison by $t_{\pi_v}>0$ and the
second by $t_v^+>0$, we obtain
\[
 (1-\nu)
 \left(\frac1{t_{\pi_v}}+\frac1{t_v^+}\right)
 \le
 \frac1{\bar t_{\pi_v}}+\frac1{\bar t_v^+}
 \le
 (1+\nu)
 \left(\frac1{t_{\pi_v}}+\frac1{t_v^+}\right),
\]
which is precisely
\[
 (1-\nu)h_{\pi_v}
 \le\bar h_{\pi_v}
 \le(1+\nu)h_{\pi_v}.
\]

Thus the metric comparisons hold on every logical edge.
Since $\nu=K^{-4}\le1/2$,
\[
 h_e\le\frac{\bar h_e}{1-\nu}\le2\bar h_e.
\]
Combining this with the previously established gradient bound gives
\[
 |g_e-\bar g_e|
 \le\nu h_e
 \le\frac{\nu}{1-\nu}\bar h_e
 \le2\nu\bar h_e
 \qquad(e\in\mathcal E_i).
\]
\end{proof}

\subsubsection{The stopping condition certifies all labels}

\begin{lemma}\label{sec4:duality}
Let $B$ be the tail-minus-head incidence matrix of a directed
multigraph, let $g$ be an edge vector, let $h>0$ and $a\ge0$.
If
\[
 |g^\top\Delta|\le a\|h\Delta\|_1
 \qquad(B^\top\Delta=0),
\]
then there exists a vertex vector $p$, normalized by $p_s=0$,
such that
\begin{equation}\label{sec4:coordinate-certificate}
 |g_e+(Bp)_e|\le ah_e\qquad(e\in\mathcal E_i).
\end{equation}
\end{lemma}

\begin{proof}
For every edge $e=u\to v$, construct two auxiliary arcs:
$e^+:u\to v$ of length $g_e+ah_e$ and
$e^-:v\to u$ of length $-g_e+ah_e$.
If a directed cycle uses these arcs with multiplicities
$k_e^+,k_e^-$, then $\Delta_e=k_e^+-k_e^-$ is a circulation.
The cycle's length is
\[
 g^\top\Delta+a\sum_eh_e(k_e^++k_e^-)
 \ge g^\top\Delta+a\|h\Delta\|_1\ge0.
\]
Thus there is no negative directed cycle.
Add a new vertex $q$ with zero-length arcs to every vertex, and let
$d(v)$ be the finite shortest-path distance from $q$ to $v$. The two auxiliary
arcs for $e=u\to v$ give
\[
 d(v)\le d(u)+g_e+ah_e,\qquad
 d(u)\le d(v)-g_e+ah_e.
\]
Hence $p_v=d(v)-d(s)$ satisfies the desired inequalities.
\end{proof}

\begin{lemma}\label{sec4:stable-labels}
Consider any state satisfying \Cref{sec4:cache} and
\eqref{sec4:stopping-certificate}, with
$0<\theta_v\le1/4$ and $u_v=(\tau\theta_v)^{-1}$.
There exists a vertex vector $p$, with $p_s=0$, such that
\begin{equation}\label{sec4:distance-certificate}
 0<p_v\le z_i(v)\le(1+\epsilon/16)p_v
 \qquad(v\ne s).
\end{equation}
Every probe satisfies
\[
 0<x_{\pi_v}<u_v-\kappa.
\]
Moreover, for every $v\ne s$,
\[
 \begin{aligned}
 \bar x_{\pi_v}\ge u_v/2
 &\ \Longrightarrow\
 z_i(v)\le(1+4\tau)^2\theta_v,\\
 \bar x_{\pi_v}<u_v/2
 &\ \Longrightarrow\
 p_v\ge(1-4\tau)\theta_v.
 \end{aligned}
\]
\end{lemma}

\begin{proof}
For every circulation $\Delta$, the stopping certificate and
\Cref{sec4:cache} imply
\[
 \begin{aligned}
 |g^\top\Delta|
 &\le(K^{-1}+2K^{-4})\|\bar h\Delta\|_1\\
 &\le(1+K^{-4})(K^{-1}+2K^{-4})\|h\Delta\|_1\\
 &\le\tfrac1{16}\|h\Delta\|_1.
 \end{aligned}
\]
By \Cref{sec4:duality}, there exists $p$ with $p_s=0$ such that
\[
 |g_e+(Bp)_e|\le h_e/16
 \qquad(e\in\mathcal E_i).
\]

For $e\in\widehat E_i$, put $\sigma_e=c_e+(Bp)_e$.
Substituting \eqref{sec4:edge-coefficients} gives
\begin{equation}\label{sec4:core-certificate}
 \frac{15}{16t_e}\le\sigma_e\le\frac{17}{16t_e}.
\end{equation}
For an ordinary return $r=v\to s$, the same substitution gives
\begin{equation}\label{sec4:return-certificate}
 \frac{A-\sqrt A/16}{t_r}\le p_v
 \le\frac{A+\sqrt A/16}{t_r}.
\end{equation}
Thus $\sigma_e>0$ on every core edge and $p_v>0$ for every
$v\ne s$. Summing the inequalities
$p_v-p_u<c_e$ along a shortest core path from $s$ to $v$ gives
$p_v\le z_i(v)<1/8$.

For a core edge $e=u\to v$,
\[
 0<\sigma_e=c_e+p_u-p_v
 <\frac1{8n}+\frac18<\frac14.
\]
Hence \eqref{sec4:core-certificate} implies $t_e>15/4>1$.
For an ordinary return, \eqref{sec4:return-certificate} implies
\[
 t_r\ge\frac{A-\sqrt A/16}{p_v}
 >8(A-\sqrt A/16)>1.
\]
Therefore every base edge satisfies
\[
 x_e=t_e-\kappa>0,
 \qquad
 (1-\kappa)t_e\le x_e\le t_e.
\]
Using $\kappa<1/16$ and $\sqrt A\le A$ in
\eqref{sec4:core-certificate}--\eqref{sec4:return-certificate}
now gives
\begin{equation}\label{sec4:flow-comparisons}
 \frac1{2\sigma_e}\le x_e\le\frac2{\sigma_e}
 \quad(e\in\widehat E_i),\qquad
 \frac A{2p_v}\le x_r\le\frac{2A}{p_v}
 \quad(r=v\to s\text{ ordinary}).
\end{equation}
Since $v$ has exactly $q_i(v)$ ordinary returns,
\begin{equation}\label{sec4:ordinary-return-sums}
 \frac{Aq_i(v)}{2p_v}
 \le\sum_{\substack{r\in\mathcal B_i\\r=v\to s}}x_r
 \le\frac{2Aq_i(v)}{p_v}.
\end{equation}

For a probe $\pi_v$, the coordinate certificate gives
\[
 \left|
 p_v-\theta_v-\frac1{t_{\pi_v}}+\frac1{t_v^+}
 \right|
 \le\frac1{16}
 \left(\frac1{t_{\pi_v}}+\frac1{t_v^+}\right).
\]
Equivalently,
\[
 \frac{15}{16t_{\pi_v}}-\frac{17}{16t_v^+}
 \le p_v-\theta_v\le
 \frac{17}{16t_{\pi_v}}-\frac{15}{16t_v^+}.
\]
We have $u_v\ge256$ and $-1/4<p_v-\theta_v<1/8$.
If $x_{\pi_v}\le0$, then $t_{\pi_v}=\kappa$ and $t_v^+\ge u_v$,
so
\[
 p_v-\theta_v
 \ge\frac{15}{16\kappa}-\frac{17}{16u_v}>\frac18,
\]
a contradiction.
If $x_{\pi_v}\ge u_v-\kappa$, then $t_v^+=\kappa$ and
$t_{\pi_v}\ge u_v$, so
\[
 p_v-\theta_v
 \le\frac{17}{16u_v}-\frac{15}{16\kappa}<-\frac14,
\]
again a contradiction. Consequently,
\[
 0<x_{\pi_v}<u_v-\kappa,\qquad
 t_{\pi_v}=x_{\pi_v}+\kappa,\qquad
 t_v^+=u_v-x_{\pi_v}.
\]

For $R\subseteq V\setminus\{s\}$, let $I(R)$ and $O(R)$
denote the total base flows entering and leaving $R$.
Every probe incident to $R$ leaves $R$, so $B^\top x=0$ gives
\begin{equation}\label{sec4:base-cut-balance}
 I(R)=O(R)+\sum_{v\in R}x_{\pi_v}\ge O(R).
\end{equation}

Fix a target $t\ne s$. Give each core edge the positive length
$\ell_e=\sigma_e$, and set
\[
 a_u=\operatorname{dist}_\ell(u,t),\qquad
 R(r)=\{u:a_u\le r\},\qquad
 r_0=\epsilon p_t/16.
\]
Suppose $s\notin R(r_0)$.
For $0\le r\le r_0$ and $u\in R(r)$, let $P_{u,t}$ be a shortest
$u$--$t$ path in the metric $\ell$. Telescoping gives
\begin{equation}\label{sec4:ball-labels}
 p_u-p_t
 =\sum_{e\in P_{u,t}}(\sigma_e-c_e)
 \le a_u\le r.
\end{equation}
Define
\begin{equation}\label{sec4:volume}
 V_t(r)=1+\sum_{e=u\to v\in\widehat E_i}
 \min\left\{1,\max\left\{0,\frac{r-a_v}{\ell_e}\right\}\right\},
\end{equation}
where a summand with $a_v=\infty$ is zero.
Then $V_t(0)=1$ and $V_t(r)\le2m+1$.
Every nonzero summand has its head in $R(r)$.
Since $t\in R(r)$ and $q_i(t)\ge1$,
\begin{equation}\label{sec4:volume-indegrees}
 V_t(r)\le1+\sum_{v\in R(r)}q_i(v)
 \le2\sum_{v\in R(r)}q_i(v).
\end{equation}

For almost every $r\in[0,r_0]$, each core edge $u\to v$ entering
$R(r)$ satisfies
\[
 a_v\le r<a_u\le a_v+\ell_e
\]
and contributes $1/\ell_e$ to $V_t'(r)$.
No ordinary return enters $R(r)$ because $s\notin R(r)$.
Using \eqref{sec4:flow-comparisons},
\eqref{sec4:ordinary-return-sums},
\eqref{sec4:base-cut-balance}, and
\eqref{sec4:ball-labels}--\eqref{sec4:volume-indegrees}, we obtain
\[
 \begin{aligned}
 V_t'(r)
 &\ge\tfrac12 I(R(r))
 \ge\tfrac12 O(R(r))\\
 &\ge\frac A4\sum_{v\in R(r)}\frac{q_i(v)}{p_v}\\
 &\ge\frac A{4(p_t+r)}
       \sum_{v\in R(r)}q_i(v)\\
 &\ge\frac A{8(p_t+r)}V_t(r).
 \end{aligned}
\]
Since $V_t(r)\ge1$, integration yields
\[
 \begin{aligned}
 \log V_t(r_0)
 &=\int_0^{r_0}\frac{V_t'(r)}{V_t(r)}\,dr\\
 &\ge\int_0^{r_0}\frac A{8(p_t+r)}\,dr\\
 &\ge\frac{Ar_0}{8(p_t+r_0)}
 =\frac{A\epsilon}{128(1+\epsilon/16)}
 >\log(2m+1).
 \end{aligned}
\]
The last inequality follows from the choice of $A$.
This contradicts $V_t(r_0)\le2m+1$.

Thus there exists an $s$--$t$ core path $P$ with
$\sum_{e\in P}\sigma_e\le r_0$.
Telescoping its costs gives
\[
 z_i(t)\le\sum_{e\in P}c_e
 =p_t+\sum_{e\in P}\sigma_e
 \le(1+\epsilon/16)p_t.
\]
Together with $0<p_t\le z_i(t)$, this proves
\eqref{sec4:distance-certificate}.

It remains to prove the two flag implications.
Write $\nu=K^{-4}$.
If $\bar x_{\pi_v}\ge u_v/2$, then
$\bar t_{\pi_v}=\bar x_{\pi_v}+\kappa$, and the cache bound gives
\[
 x_{\pi_v}
 \ge\bar x_{\pi_v}-\nu(\bar x_{\pi_v}+\kappa)
 \ge(1-\nu)u_v/2-\nu\kappa
 \ge u_v/3.
\]
Therefore
\[
 \begin{aligned}
 p_v-\theta_v
 &\le\frac{17}{16(x_{\pi_v}+\kappa)}
       -\frac{15}{16(u_v-x_{\pi_v})}\\
 &\le\frac{51}{16u_v}\le\frac4{u_v}.
 \end{aligned}
\]
Since $u_v^{-1}=\tau\theta_v$ and
$\epsilon/16=4\tau$,
\[
 z_i(v)\le(1+4\tau)p_v
 \le(1+4\tau)\left(\theta_v+\frac4{u_v}\right)
 =(1+4\tau)^2\theta_v.
\]

If $\bar x_{\pi_v}<u_v/2$, then
$\bar t_v^+=u_v-\bar x_{\pi_v}$, and the cache bound gives
\[
 x_{\pi_v}
 \le\bar x_{\pi_v}+\nu(u_v-\bar x_{\pi_v})
 \le(1+\nu)u_v/2
 \le2u_v/3.
\]
Consequently,
\[
 \begin{aligned}
 p_v-\theta_v
 &\ge\frac{15}{16(x_{\pi_v}+\kappa)}
       -\frac{17}{16(u_v-x_{\pi_v})}\\
 &\ge-\frac{51}{16u_v}\ge-\frac4{u_v},
 \end{aligned}
\]
and hence
\[
 p_v\ge\theta_v-\frac4{u_v}
 =(1-4\tau)\theta_v.
\]
\end{proof}

\subsubsection{Descent of the circulation potential}

\begin{lemma}\label{sec4:descent}
With the graph and all probe parameters fixed, every successful
augmentation decreases $\Phi$ by at least $K^{-6}/2$.
This holds for the potential including both barriers of every probe.
\end{lemma}

\begin{proof}
Let $\Delta$ be the added circulation. All coefficients and probe
parameters are fixed during this step. For $0\le\lambda\le1$, define
\[
 F(\lambda)=\Phi(x+\lambda\Delta),
 \qquad
 t_e(\lambda)=\max\{x_e+\lambda\Delta_e+\kappa,\kappa\}
 \quad(e\in\mathcal E_i),
\]
and
\[
 t_v^+(\lambda)
 =\max\{u_v-x_{\pi_v}-\lambda\Delta_{\pi_v},\kappa\}
 \quad(v\ne s).
\]
Since $\bar h_e\ge1/\bar t_e$, \Cref{sec4:cache} and
\eqref{sec4:augmentation} give
\[
 \left|\frac{t_e(\lambda)}{\bar t_e}-1\right|
 \le\frac{|x_e-\bar x_e|+|\Delta_e|}{\bar t_e}
 \le2K^{-4}
 \qquad(e\in\mathcal E_i).
\]
Similarly, using $\bar h_{\pi_v}\ge1/\bar t_v^+$,
\[
 \left|\frac{t_v^+(\lambda)}{\bar t_v^+}-1\right|
 \le\frac{|x_{\pi_v}-\bar x_{\pi_v}|+|\Delta_{\pi_v}|}{\bar t_v^+}
 \le2K^{-4}
 \qquad(v\ne s).
\]
Consequently,
\[
 t_e(\lambda)\ge(1-2K^{-4})\bar t_e,
 \qquad
 t_v^+(\lambda)\ge(1-2K^{-4})\bar t_v^+.
\]
For almost every $\lambda\in[0,1]$,
\[
 \begin{aligned}
 F''(\lambda)
 &\le\sum_{e\in\mathcal B_i}
       \frac{\mu_e\Delta_e^2}{t_e(\lambda)^2}
 +\sum_{v\ne s}\Delta_{\pi_v}^2
  \left(\frac1{t_{\pi_v}(\lambda)^2}
              +\frac1{t_v^+(\lambda)^2}\right)\\
 &\le\frac1{(1-2K^{-4})^2}
 \left[
  \sum_{e\in\mathcal B_i}\frac{\mu_e\Delta_e^2}{\bar t_e^2}
  +\sum_{v\ne s}\Delta_{\pi_v}^2
   \left(\frac1{\bar t_{\pi_v}^2}+\frac1{(\bar t_v^+)^2}\right)
 \right]\\
 &\le\frac1{(1-2K^{-4})^2}
       \sum_{e\in\mathcal E_i}(\bar h_e\Delta_e)^2.
 \end{aligned}
\]
The last inequality uses
\[
 \bar h_e^2=\frac{\mu_e}{\bar t_e^2}
 \quad(e\in\mathcal B_i),\qquad
 \bar h_{\pi_v}^2
 =\left(\frac1{\bar t_{\pi_v}}+\frac1{\bar t_v^+}\right)^2
 \ge\frac1{\bar t_{\pi_v}^2}+\frac1{(\bar t_v^+)^2}.
\]
Since $\psi_\kappa'$ is Lipschitz, $F'$ is absolutely continuous.
The integral Taylor formula therefore gives
\[
 \begin{aligned}
 \Phi(x+\Delta)-\Phi(x)
 &=g^\top\Delta+\int_0^1(1-\lambda)F''(\lambda)\,d\lambda\\
 &\le g^\top\Delta+
 \frac1{2(1-2K^{-4})^2}\sum_e(\bar h_e\Delta_e)^2\\
 &\le g^\top\Delta+2\sum_e(\bar h_e\Delta_e)^2\\
 &\le\bar g^\top\Delta+2K^{-4}\|\bar h\Delta\|_1
       +2\|\bar h\Delta\|_1^2\\
 &\le-K^{-6}+4K^{-8}\le-\tfrac12K^{-6}.
 \end{aligned}
\]
\end{proof}

\subsubsection{A global budget}

\begin{lemma}[Threshold invariant]\label{sec4:thresholds}
At every finite execution prefix,
\[
 z_i(v)\le(1+\epsilon/4)\theta_v\qquad(v\ne s).
\]
If $J$ is the number of threshold transactions performed so far, then
\[
 J\le\frac{mA\log M}{64}.
\]
\end{lemma}

\begin{proof}
Treat each threshold transaction as one event.
We prove the first claim by induction over finite execution prefixes.
By \Cref{sec4:stable-labels}, a flagged probe satisfies, with its
parameters before the transaction,
\[
 z_i(v)\le(1+4\tau)^2\theta_v
 \le(1+16\tau)(1-\tau)\theta_v,
\]
because
\[
 (1+16\tau)(1-\tau)-(1+4\tau)^2
 =7\tau-32\tau^2\ge0.
\]
Since $1+16\tau=1+\epsilon/4$, replacing $\theta_v$ by
$(1-\tau)\theta_v$ preserves the invariant.
The \textsc{QueryFlow} call changes neither the physical flow nor the
supplied coefficients, so the certificate and flag from before that
call authorize the step.
The invariant holds initially because $z_0(v)<1/8<\theta_v=1/4$.
An insertion only decreases $z_i$, while augmentations and ordinary
refreshes change neither $z_i$ nor the thresholds.
This induction does not assume termination of repair.

Since $z_i(v)\ge1/Q$, the invariant implies
$\theta_v\ge1/((1+\epsilon/4)Q)$.
If $j_v$ counts the transactions at $v$, then
$\theta_v=\frac14(1-\tau)^{j_v}$ and
\[
 j_v\le
 \frac{\log((1+\epsilon/4)Q/4)}{-\log(1-\tau)}
 \le\tau^{-1}\log(2Q).
\]
Therefore
\[
 J\le(n-1)\tau^{-1}\log(2Q)
 \le\frac{mA\log M}{64},
\]
where the last inequality uses $A\ge64\tau^{-1}$,
$n-1\le m$, and $M\ge2Q$.
\end{proof}

\begin{lemma}[Global budget]\label{sec4:budget}
Set
\begin{equation}\label{sec4:energy}
 \mathcal B=100mA\log M.
\end{equation}
The total potential decrease due to successful augmentations is
at most $2\mathcal B$.
Throughout repair, including all augmentation segments,
\begin{equation}\label{sec4:ranges}
 \begin{gathered}
 -4\kappa\mathcal B\le x_e,\bar x_e,\qquad
 \kappa\le t_e,\bar t_e\le M^{10}
 \quad(e\in\mathcal E_i),\\
 \kappa\le t_v^+,\bar t_v^+\le M^{10}
 \quad(v\ne s).
 \end{gathered}
\end{equation}
Further, every call to \textsc{Repair} terminates.
\end{lemma}

\begin{proof}
By \Cref{sec4:thresholds} and $z_i(v)\ge1/Q$,
\[
 \frac1{(1+\epsilon/4)Q}\le\theta_v\le\frac14,
 \qquad
 256\le u_v=(\tau\theta_v)^{-1}
 \le\frac{2Q}{\tau}\le M^2.
\]
Here $\theta_v\le1/4$ follows from its multiplicative decreases,
and $2Q/\tau\le M^2$ follows from the parameter choices.
The transaction bound of \Cref{sec4:thresholds} also holds at every
finite prefix. We use these bounds without assuming termination.

Take the fixed reference vector
$p_s^0=0$, $p_v^0=1/(2Q)$ for $v\ne s$, and put
$a_e=c_e+(Bp^0)_e$.
On base edges, $1/(2Q)\le a_e\le2$.
Conservation gives $c^\top x=a^\top x$, so
\[
 \Phi(x)=\sum_{e\in\mathcal E_i}F_e(x_e),
\]
where
\[
 \begin{aligned}
 F_e(y)&=a_ey+\mu_e\psi_\kappa(y+\kappa)
                    &&(e\in\mathcal B_i),\\
 F_{\pi_v}(y)&=(1/(2Q)-\theta_v)y
       +\psi_\kappa(y+\kappa)+\psi_\kappa(u_v-y)+\log u_v.
 \end{aligned}
\]
Each base scalar function is coercive, and its minimum is attained
at $y+\kappa=\mu_e/a_e>\kappa$, with value
\begin{equation}\label{sec4:scalar-minimum}
 b_e=\mu_e\left(1-\log\frac{\mu_e}{a_e}\right)-\kappa a_e.
\end{equation}
There are at most $4m$ base edges and their coefficient mass is at
most $4mA$. Since $\mu_e/a_e\le2QA\le M^2$,
\[
 \sum_{e\in\mathcal B_i}b_e
 \ge-8mA\log M-8m\kappa.
\]

For a probe, define the global floor
\[
 b_{\pi_v}=-\tau^{-1}-\log(u_v+\kappa).
\]
For $0\le y\le u_v$, the linear term is at least
$-\theta_vu_v=-1/\tau$, the lower barrier is at least
$-\log(u_v+\kappa)$, and the normalized upper barrier is nonnegative.
For $y<0$, writing $z=|y|$, the exact expression is
\[
 F_{\pi_v}(-z)=-\log\kappa
   +z(\kappa^{-1}-1/(2Q)+\theta_v)-\log(1+z/u_v).
\]
Its derivative in $z$ is at least $1/(2\kappa)$, so it is positive.
For $y\ge u_v$, the derivative is
$1/(2Q)-\theta_v-1/(y+\kappa)+1/\kappa\ge1/(2\kappa)$.
Thus $F_{\pi_v}(y)\ge b_{\pi_v}$ for every signed $y$.
Since $A\ge64\tau^{-1}$ and $u_v+\kappa\le M^3$, all scalar floors
together satisfy
\[
 \sum_{e\in\mathcal E_i}b_e\ge-12mA\log M>-\mathcal B.
\]

A zero-flow core insertion and its ordinary return add exactly
\begin{equation}\label{sec4:insertion-cost}
 (1+A)\log(1/\kappa)
\end{equation}
to the potential. There are at most $2m$ such pairs, including all
initial core edges and sentinels, so their total contribution is at
most $80mA\log M$.
Initializing the probes at zero adds $20(n-1)\log M$.
At a cycle-stationary state, \Cref{sec4:stable-labels} gives
$0<x_{\pi_v}<u_v-\kappa$.
In a threshold transaction the linear term increases by
$\tau\theta_vx_{\pi_v}<1$.
Its lower barrier is unchanged, while its normalized upper barrier
$-\log(1-x_{\pi_v}/u_v)$ decreases as $u_v$ increases.
Thus all transactions add less than $J$ potential.
The total positive injection is therefore at most
\[
 80mA\log M+20(n-1)\log M+J
 <82mA\log M<\mathcal B.
\]
Together with the global floor and \Cref{sec4:descent}, this proves
the total decrease bound and $-\mathcal B\le\Phi\le\mathcal B$ at
all operation endpoints. Along augmentation segments, conservation
is preserved and convexity gives the same upper bound.
Consequently,
\[
 0\le F_e(x_e)-b_e
 \le\sum_f(F_f(x_f)-b_f)
 =\Phi(x)-\sum_f b_f\le2\mathcal B.
\]

For a base edge and $y\ge0$, put $w=a_e(y+\kappa)/\mu_e$.
Then $F_e(y)-b_e=\mu_e(w-1-\log w)$.
For $w\ge2$, use $w-1-\log w\ge w/2-1$; the case $w<2$ is
immediate. In both cases,
\[
 y+\kappa\le2Q(4\mathcal B+2A)\le M^{10}.
\]
For $y<0$,
\[
 F_e(y)-b_e
 =F_e(0)-b_e+(\mu_e/\kappa-a_e)|y|
 \ge\frac{\mu_e}{2\kappa}|y|,
\]
so $y\ge-4\kappa\mathcal B/\mu_e\ge-4\kappa\mathcal B$.
The positive and negative probe tails have slope magnitude at
least $1/(2\kappa)$ beyond $u_v$ and below zero, respectively.
Their floor bounds similarly give
\[
 -4\kappa\mathcal B\le x_{\pi_v}\le u_v+4\kappa\mathcal B.
\]
Since $u_v\le M^2$, both probe widths are at most $M^{10}$.
Cached coordinates are previously attained exact coordinates;
whenever a probe parameter changes, its sample and both cached
widths are recomputed immediately. Thus the bounds also hold for
all cached quantities.

In particular, on all logical edges,
\[
 M^{-10}\le h_e,\bar h_e\le A/\kappa,
 \qquad |g_e|,|\bar g_e|\le1+A/\kappa.
\]
Because $W=m^{O(1)}$, $\epsilon^{-1}\le m$, and $K=m^{o(1)}$,
these quantities, $x,\bar x$, and the parameters have polynomial
numerical ranges. The solver parameters satisfy
$r^{-1},\Gamma^{-1},\xi^{-1}=m^{o(1)}$.
The argument applies at every finite prefix: a first exit from
the chosen ranges would contradict the bounds just proved for
the operation causing that exit. Finally the total decrease bound
and \Cref{sec4:descent} bound the number of successful steps, while
\Cref{sec4:thresholds} excludes infinitely many transactions.
Every failed call either returns from repair or starts a transaction,
so every repair terminates.
\end{proof}

\subsubsection{The accuracy of the probe-threshold readout}

\begin{lemma}\label{sec4:readout}
At every stable state and for every $v\ne s$,
\begin{equation}\label{sec4:readout-accuracy}
 Qz_i(v)\le E(v)\le(1+\epsilon/3)Qz_i(v).
\end{equation}
The array $E$ changes only at initialization and threshold
transactions, and is nonincreasing thereafter.
The public array $D$ is also nonincreasing. At every stable state it
identifies unreachability exactly and satisfies
$d_i(v)\le D(v)\le(1+\epsilon)d_i(v)$ for every vertex.
\end{lemma}

\begin{proof}
The threshold invariant in \Cref{sec4:thresholds} gives
$z_i(v)\le E(v)/Q$.
At a stable state all probes are unflagged, so
\Cref{sec4:stable-labels} gives
$p_v\ge(1-4\tau)\theta_v$.
Together with $p_v\le z_i(v)$, this yields
\begin{equation}\label{sec4:threshold-readout-comparison}
 z_i(v)\le\frac{E(v)}Q
 \le\frac{1+\epsilon/4}{1-\epsilon/16}z_i(v)
 \le(1+\epsilon/3)z_i(v).
\end{equation}
The last inequality holds for $0<\epsilon\le1$.
If $v$ is reachable, \eqref{sec4:sentinel-distances} implies
$E(v)\le(1+\epsilon)d_i(v)<H/2$.
If $v$ is unreachable, $E(v)\ge Qz_i(v)=H$.
Thus \eqref{sec4:public-definition} classifies reachability exactly
and supplies the required distance guarantee; the source entry is
zero by definition.
Finally only a threshold transaction changes a nonsource entry,
and every such change multiplies $E(v)$ by $1-\tau$.
The threshold map defining $D$ preserves this monotonicity.
\end{proof}

\subsection{Total Running Time}\label{sec4:time}

Let $T$ be the number of successful augmentations, including
initialization repair, and let $J$ be the number of threshold
transactions. \Cref{sec4:descent,sec4:budget} give
\begin{equation}\label{sec4:augmentation-count}
 T\le4\mathcal B K^6=m^{1+o(1)}\epsilon^{-1}.
\end{equation}
Moreover, \Cref{sec4:thresholds} gives $J=O(mA\log M)$.
Each failed call either finishes a repair or authorizes a transaction.
Hence the total number $L$ of cycle calls is at most
$T+m+J+1$. Each transaction invokes one \textsc{QueryFlow} call
and one \textsc{UpdateEdge} call not used to process ordinary
notifications, so $R=F=J$.
Using \eqref{sec4:solver-parameters}, the maintenance
bound of \Cref{prop:cycle-interface} is
\begin{equation}\label{sec4:notification-count}
 O\!\left(
 \left(m+L+F+R+\frac{T\Gamma}{\xi r}\right)m^{o(1)}
 \right)
 =O((m+T+J)m^{o(1)}).
\end{equation}
This includes all notifications, \textsc{QueryFlow} calls, and
\textsc{UpdateEdge} calls, including those not used to process
ordinary notifications.
Maintaining the flag list takes constant explicit work per refresh.
Only initialization and the $J$ threshold transactions write
distance entries. Therefore the total maintenance time is
\begin{equation}\label{sec4:total-time}
 O((m+T+J)m^{o(1)})
 =O(m^{1+o(1)}A\log M)
 =m^{1+o(1)}\epsilon^{-1}.
\end{equation}
The number of explicit distance writes is $O(n+J)$ and satisfies
the same asymptotic bound.

\begin{proof}[Proof of \Cref{sec4:main}]
\Cref{sec4:budget} proves termination of every repair, and
\Cref{sec4:readout} proves the distance and reachability guarantees
at every state in which queries are served.
The total maintenance bound is \eqref{sec4:total-time}.
\end{proof}

\section{Incremental Path Reporting}\label{sec5:paths}

In this section, we establish the following lemma to extend our result
from distance estimates to path reporting.

\begin{lemma}\label{lem:path-reporting}
Let $G=(V,E)$ be an initially empty directed graph on $n\ge2$ vertices,
undergoing at most $m\ge n$ edge insertions with weights in $[1,W]$,
and let $s\in V$ be a fixed source. Fix $\delta\in(0,1)$ and write
\[
 d_t(v)=\dist_{G_t}(s,v),
 \qquad H=\lceil\log_2(m+1)\rceil.
\]
Suppose an oracle maintains nonincreasing upper estimates $a_t(v)$,
with $a_t(s)=0$.
For each $v\ne s$, $a_t(v)=\infty$ if and only if $v$ is
unreachable. Every finite nonsource estimate is an integral power of $1+\delta$
and satisfies
\[
 a_t(v)\in\{(1+\delta)^k:k\in\mathbb Z_{\ge0}\},
 \qquad
 d_t(v)\le a_t(v)\le(1+\delta)^2d_t(v).
\]
After each insertion, the oracle reports precisely the vertices
$v\ne s$ whose estimates have decreased, including those that have
just become reachable.

There is a deterministic data structure with total maintenance time
\[
 O\!\left(\delta^{-1}m\log m\log n\log(nW)\right)
\]
that maintains nonincreasing labels $\ell_t(v)$ satisfying
\[
 d_t(v)\le\ell_t(v)
 \le(1+\delta)^{2(H+1)}d_t(v)
\]
for every reachable vertex $v$.
Given such a vertex $v$, it reports a simple $s$-to-$v$ path $P$ with
$w(P)\le\ell_t(v)$ in $O(|P|)$ time, where $|P|$ denotes the number of
vertices on $P$.
A label query and an unreachability query take
$O(1)$ time.

\end{lemma}

\subsection{Algorithm}\label{sec5:algorithm}

The algorithm maintains a label $\ell(v)$ and, whenever $\ell(v)$ is
finite and $v\ne s$, a parent edge $(p(v),v)$.
Initially $\ell(s)=0$, every other label is infinite, and all parent
pointers are undefined.

Assign each vertex a level $h(v)\in\{0,\ldots,H\}$. For $0\le i<H$, maintain the bucket
\[
 A_i=\{v\ne s:h(v)=i\}.
\]
The buckets are pairwise disjoint. Level $H$ is implicit, and $s$ always remains at that level. Initially all vertices have
level $H$, so all explicit buckets are empty.

\paragraph{Processing a batch.}
Store the newly inserted edge and let $B$ be the batch of reported
vertices. If $B$ is empty, no repair is needed.
Otherwise, first remove every vertex of $B$ from its current explicit
bucket, if any. Find the smallest index $k\in\{0,\ldots,H-1\}$ whose
bucket is now empty. Such an index always exists, as proved below.
Form
\begin{equation}\label{sec5:repair-set}
 S=B\cup\bigcup_{i=0}^{k-1}A_i.
\end{equation}
Clear the buckets below $k$, perform a local repair (defined below) on $S$, and put
all vertices of $S$ into $A_k$, setting their levels to $k$.
All other buckets and levels remain unchanged.

\paragraph{Local repair.}
For a set $S\subseteq V\setminus\{s\}$, we repair the labels of its
vertices using the current labels outside $S$ as boundary values.
Construct an auxiliary graph on $S\cup\{r\}$, where $r$ is a new
auxiliary source, and retain all original edges with both endpoints
in $S$. For each $v\in S$, define
\begin{equation}\label{sec5:boundary}
 b(v)=\min_{\substack{u\notin S\\(u,v)\in E}}
       \{\ell(u)+w(u,v)\},
\end{equation}
where an empty minimum is infinity.
If $b(v)$ is finite, add an edge $r\to v$ of weight $b(v)$ and store
an original boundary edge attaining this minimum.

Run Dijkstra's algorithm from $r$ in the auxiliary graph, obtaining
distances $D(v)$ and shortest-path-tree parents.
Using the labels from before the repair to define the boundary values,
set
\begin{equation}\label{sec5:local-update}
 \ell'(v)=\min\{\ell(v),D(v)\},
 \qquad v\in S.
\end{equation}
If the label of $v$ strictly decreases, replace its parent by the
auxiliary tree parent. A tree edge $r\to v$ is translated into its
stored original boundary edge; every other tree edge is already an
original edge. If the label does not decrease, keep the previous parent.
Labels and parents outside $S$ remain unchanged. All these changes
are completed before any query is answered.

\paragraph{Queries.}
A label query returns $\ell(v)$. If $\ell(v)=\infty$, a path query
reports unreachability. Otherwise, follow parent edges from $v$ to
$s$, collect their original identifiers, and reverse the list.

\subsection{Correctness analysis}\label{sec5:correctness}

\paragraph{Existence of an empty bucket.}
Consider the potential
\[
 \Phi=\sum_{i=0}^{H-1}2^i\mathbf{1}[A_i\ne\varnothing].
\]
We prove inductively that every nonempty batch admits a local repair
and that $\Phi\le j$ after the first $j$ nonempty batches.
Initially $\Phi=0$.
Before processing the $j$-th nonempty batch, $\Phi\le j-1$ by induction,
and removing its reported vertices can only decrease $\Phi$.
If all explicit buckets were then nonempty, we would have
\[
 \Phi=2^H-1\ge m>j-1,
\]
a contradiction, since there are at most $m$ nonempty batches.
Thus the smallest empty bucket $k$ exists.
All buckets below $k$ are nonempty, so clearing them and filling $A_k$
increases the potential by exactly
\[
 2^k-\sum_{i=0}^{k-1}2^i=1.
\]
Consequently $\Phi\le j$ after the repair, completing the induction.

\paragraph{Parents, monotonicity, and reachability.}

Let $\alpha=(1+\delta)^2$. $\ell(s)$ is always zero. For every finite-label $v\ne s$, we prove the parent invariant
\begin{equation}\label{sec5:parent-invariant}
 \ell(p(v))+w(p(v),v)\le\ell(v).
\end{equation}
We also show that a vertex has a finite label if and only if it is
reachable. All invariants refer to completed update batches.

Initially only the source has a finite label, so the parent invariant
holds. An edge insertion by itself preserves all existing parent
edges and labels. Consider a local repair.

For a vertex $v$ whose label decreases, we have $\ell'(v)=D(v)$.
If its auxiliary tree parent is $u\in S$, then
\[
 \ell'(u)+w(u,v)
 \le D(u)+w(u,v)
 =D(v)
 =\ell'(v).
\]
If its auxiliary parent is $r$, the stored boundary edge $u\to v$
satisfies $u\notin S$ and
\[
 \ell'(u)+w(u,v)
 =\ell(u)+w(u,v)
 =b(v)
 =D(v)
 =\ell'(v).
\]
For an unchanged finite-label vertex, its parent remains the same
and the parent's label can only decrease. Thus local repairs preserve
\eqref{sec5:parent-invariant}.

Positive edge weights imply that labels strictly decrease along parent
pointers toward the source. Parent pointers therefore contain no
cycle. Every finite-label nonsource vertex has a finite-label parent,
so its parent chain must terminate at $s$.
The resulting path is simple, and summing
\eqref{sec5:parent-invariant} along it gives
\begin{equation}\label{sec5:path-bound}
 d_t(v)\le w(P)\le\ell_t(v).
\end{equation}
In particular, an unreachable vertex cannot have a finite label.

It remains to show that every reachable vertex receives a finite
label. Suppose this holds before an insertion.
Every newly reachable nonsource vertex is reported and hence belongs
to the repair set.
For a reachable $v$ in a local repair set $S$, take a current shortest
path to $v$ and let $u$ be its last vertex outside $S$. Such a vertex
exists because $s\notin S$. If $u\ne s$, it was already reachable
before the insertion, since every newly reachable vertex belongs to
$S$. Its boundary label is therefore finite by induction.
The boundary edge after $u$ and the remaining suffix are represented
in the auxiliary graph, implying $D(v)<\infty$.
Reachable vertices outside $S$ already had finite labels.

Following a parent chain and reversing the resulting edge list takes
$O(|P|)$ time. The label and reachability tests use constant time.

\paragraph{Approximation.}
For every finite-label $v\ne s$, we prove the level invariant
\begin{equation}\label{sec5:level-invariant}
 \ell_t(v)\le\alpha^{H-h(v)}a_t(v)
 \quad\text{for every reachable }v.
\end{equation}

At initialization there are no reachable nonsource vertices, and the
level invariant holds at $s$.

Consider an insertion processed by a local repair at level $k$.
By the choice of the smallest empty bucket and
\eqref{sec5:repair-set}, every vertex outside $S$ has level at least
$k+1$. Moreover, all reported vertices are in $S$, so an outside
vertex's oracle estimate has not changed in this insertion.
For every currently reachable $u\notin S$, the previous level
invariant and the oracle guarantee therefore give
\begin{equation}\label{sec5:outside-bound}
 \begin{aligned}
 \ell(u)
 &\le\alpha^{H-h(u)}a_{t-1}(u)\\
 &=\alpha^{H-h(u)}a_t(u)\\
 &\le\alpha^{H-k-1}a_t(u)
 \le\alpha^{H-k}d_t(u).
 \end{aligned}
\end{equation}
Such a vertex was already reachable before the insertion, as argued
above, so the previous invariant applies.

For $v\in S$, take a current shortest path and let $u$ be its last
vertex outside $S$. The auxiliary graph represents the boundary edge
after $u$ and all subsequent edges. Hence
\[
 \begin{aligned}
 D(v)
 &\le\ell(u)+d_t(v)-d_t(u)\\
 &\le\alpha^{H-k}d_t(u)+d_t(v)-d_t(u)\\
 &\le\alpha^{H-k}d_t(v)
 \le\alpha^{H-k}a_t(v).
 \end{aligned}
\]
The new level of $v$ is $k$, and $\ell'(v)\le D(v)$, proving
\eqref{sec5:level-invariant} for every repaired vertex.
For vertices outside $S$, their estimates, labels, and levels remain
unchanged, so the invariant is preserved there as well.

If a report batch is empty, all estimates, labels, and levels are
unchanged, and the invariant continues to hold.
Combining \eqref{sec5:level-invariant} with the oracle guarantee yields
\[
 \ell_t(v)
 \le\alpha^{H-h(v)}a_t(v)
 \le\alpha^{H+1}d_t(v)
 =(1+\delta)^{2(H+1)}d_t(v).
\]
Together with \eqref{sec5:path-bound}, this proves the approximation
and path-reporting guarantees.

\subsection{Time complexity}\label{sec5:time}

Let $q_v$ be the number of reports of $v$, and let
$Q=\sum_{v\ne s}q_v$.
Every finite nonsource distance is in $[1,(n-1)W]$, so every finite
reported estimate is in $[1,(1+\delta)^2(n-1)W]$.
Successive finite estimates decrease by a factor of at least
$1+\delta$, since they are distinct powers of that number. Including
the first reachability report, we obtain
\begin{equation}\label{sec5:report-count}
 q_v\le K:=
 3+\left\lceil\frac{\log(nW)}{\log(1+\delta)}\right\rceil
 =O\!\left(\delta^{-1}\log(nW)\right).
\end{equation}
In particular, $Q\le(n-1)K$.

\paragraph{Local repairs.}
A local repair on $S$ is implemented by scanning the incoming
adjacency lists of its vertices.

The construction, Dijkstra computation, and label and parent updates
take
\begin{equation}\label{sec5:repair-cost}
    {O}\left(\sum_{v\in S}\deg_t^-(v)\log n\right)
\end{equation}
time, where $\deg_t^-(v)$ is the current indegree.

When a vertex participates without being reported in that batch, it
comes from a bucket strictly below $k$ and moves to level $k$.
Its level therefore strictly increases.
Before its first report a vertex is at the implicit level $H$ and
cannot participate in a local repair. Between successive reports, and after its last report, it can
undergo at most $H$ such strict increases. Consequently, the total
number of local repairs containing $v$ is at most $(H+1)q_v$.

Let $\deg^-_{\mathrm{fin}}(v)$ denote the final indegree of $v$.
Summing \eqref{sec5:repair-cost} gives
\[
 \begin{aligned}
 T_{\mathrm{local}}
 &=
 {O}\!\left(
 (H+1)
 \sum_{v\ne s}\deg^-_{\mathrm{fin}}(v)q_v\log n
 \right)\\
 &=
 O(mK\log m\log n),
 \end{aligned}
\]
where the last equality uses $\sum_{v\ne s}\deg^-_{\mathrm{fin}}(v)\le m$, $q_v\le K$, and $H=O(\log m)$.

\paragraph{Remaining operations.}
Initialization takes $O(m)$ time.
Removing reported vertices and finding the smallest empty bucket
take $O(QH)$ total time; moving and clearing buckets is charged to
the participating vertices above.
All these costs are absorbed by $O(mK\log m\log n)$.
Substituting \eqref{sec5:report-count} proves the claimed total bound
\[
 O\!\left(\delta^{-1}m\log m\log n\log(nW)\right).
\]

\subsection{Concluding \Cref{thm:main-2}}

For a target accuracy $0<\epsilon\le1$, set
\[
 \delta=\frac{\epsilon}{4(H+1)},
 \qquad H=\lceil\log_2(m+1)\rceil.
\]
We first adapt the distance estimates maintained by
\Cref{sec4:main} to the oracle interface required by the
path-reporting reduction.

Under the standing assumption $\epsilon^{-1}\le m$, use the size
upper bound $\widehat m=8m(H+1)$ when invoking \Cref{sec4:main}.
Then $\widehat m\ge n$ and
\[
 \delta\ge\frac1{4m(H+1)}>\frac1{\widehat m},
\]
so its accuracy condition is satisfied. Since
$\widehat m=O(m\log m)$, this changes the maintenance bound only
by factors absorbed into $m^{o(1)}$.

Run the distance data structure with accuracy parameter $\delta$.
Let $b_t(v)$ denote its public estimate after the $t$-th insertion
and the completion of the corresponding repair.
Thus $b_t(v)$ refers only to a stable state and satisfies
\[
 d_t(v)\le b_t(v)\le(1+\delta)d_t(v).
\]
For every $v\ne s$, define
\[
 \widehat b_t(v)=\min_{0\le\tau\le t}b_\tau(v),
 \qquad
 a_t(v)=(1+\delta)^{
   \left\lceil\log_{1+\delta}\widehat b_t(v)\right\rceil},
\]
with $a_t(v)=\infty$ when $\widehat b_t(v)=\infty$.
Set $a_t(s)=0$.

Since distances are nonincreasing under edge insertions,
for every $\tau\le t$ we have
$b_\tau(v)\ge d_\tau(v)\ge d_t(v)$.
Consequently,
\[
 d_t(v)\le\widehat b_t(v)\le b_t(v)
 \le(1+\delta)d_t(v),
\]
and upward rounding gives
\[
 d_t(v)\le a_t(v)\le(1+\delta)^2d_t(v).
\]
The estimates $a_t(v)$ are nonincreasing and are infinite
exactly at unreachable vertices.
Every finite nonsource distance is at least $1$, so every
finite $a_t(v)$ is a nonnegative integral power of $1+\delta$.

To produce the required reports, during each insertion and its
repair we maintain a set of vertices whose public
distance entries are written. Collecting and deduplicating these vertices
costs $O(1)$ time per write, charged to the distance data structure.
Once repair terminates, we read their final estimates,
update their historical minima and rounded estimates, and report
precisely those vertices whose rounded estimates strictly decrease.
All reports for one insertion are delivered as a single batch.
In particular, newly reachable vertices are reported.

Since every finite stable estimate is at most $(1+\delta)nW$,
a newly reachable vertex can start from a common exponent
$R=O(\delta^{-1}\log(nW))$ and decrease it until the rounded
value is found.
Thereafter its exponent only decreases.
Thus the total additional work is
$O(n\delta^{-1}\log(nW))$.

We may now apply the path-reporting reduction.
Our choice of $\delta$ gives
\[
 (1+\delta)^{2(H+1)}
 \le\exp\bigl(2\delta(H+1)\bigr)
 =\exp(\epsilon/2)
 \le1+\epsilon.
\]
Hence the reported paths are $(1+\epsilon)$-approximate,
with reporting time linear in the number of vertices on the
returned path.
The additional path-maintenance time is
\[
 O\bigl(\delta^{-1}m\log m\log n\log(nW)\bigr)
 =
 O\bigl(\epsilon^{-1}m(\log m)^2\log n\log(nW)\bigr).
\]
By \Cref{sec4:main}, the distance oracle has maintenance time
\[
 m^{1+o(1)}\delta^{-1}\log W
 =
 m^{1+o(1)}\epsilon^{-1}\log W.
\]
Including the adapter and the path-maintenance structure,
the total maintenance time remains
\[
 m^{1+o(1)}\epsilon^{-1}\log W.
\]

\section*{Acknowledgments}
We thank Yang P. Liu for pointing out that our result can be extended to incremental graphs.

\section*{AI methodology}
We used ChatGPT (OpenAI) to assist with literature research,
algorithm design, proof development and checking, and the drafting
and revision of the manuscript. This assistance materially affected
the entire paper.
The authors take full responsibility for
the correctness and originality of all content, including the
mathematical arguments and the accuracy of the references.

\bibliographystyle{alpha}
\bibliography{references}

\end{document}